\documentclass[11pt]{article}

\usepackage{authblk}
\usepackage{orcidlink}

\usepackage{amsfonts}
\usepackage{hyperref}

\usepackage[a4paper,top=1in,bottom=1in,left=1in,right=1in]{geometry}
\usepackage{comment}
\usepackage{epstopdf}
\usepackage{subfigure}
\usepackage[authoryear]{natbib}

\usepackage{amsmath}
\usepackage{amsthm}
\theoremstyle{plain}
\newtheorem{theorem}{Theorem}[section]
\newtheorem{corollary}[theorem]{Corollary}

\newtheorem{proposition}[theorem]{Proposition}

\newtheorem{assumption}{Assumption}
\newtheorem{remark}{Remark}

\usepackage{tabularx}
\usepackage{graphicx}
\usepackage{bm}
\usepackage{xcolor} 

\usepackage{titlesec}

\title{If Not Now, Then When? Model Risk in the Optimal Exercise of American Options}

\author{Luna Rigby$^{\ast}$$^{\dag}$%
\thanks{Corresponding author. Email: luna.rigby@wu.ac.at}%
\orcidlink{0009-0008-4756-2891}}
\author{Rüdiger Frey$^{\dag}$\orcidlink{0000-0002-8402-4653}}
\author{Erik Schlögl$^{\ddag \diamondsuit \clubsuit}$\orcidlink{0000-0001-8898-6879}}
\affil{\footnotesize $\dag$ WU Vienna University of Economics and Business, Austria\\
$\ddag$ School of Mathematical and Physical Sciences, University of Technology Sydney, Australia\\
$\diamondsuit$ The African Institute of Financial Markets and Risk Management (AIFMRM), University of Cape Town, South Africa\\
$\clubsuit$ Faculty of Science, Department of Statistics, University of Johannesburg, South Africa}

\begin{document}

\date{\textit{Preprint (\today)}}

\vspace{-2cm}
\maketitle
\vspace{-3.5em}

\renewcommand{\abstractname}{}
\begin{abstract}
\noindent \textbf{Abstract.} Model risk arises from the misspecification of probabilistic models used for pricing and hedging derivatives. While model risk for European-style claims has been widely studied, much less attention has been given to American-style derivatives and the associated optimal stopping problems. This paper analyzes model risk in the optimal exercise of an American put option using the benchmark methodology of \cite{Hull2002}. The true data-generating process is assumed to follow a Heston stochastic volatility model. We compare the optimal exercise strategy of an investor who correctly uses the Heston model with those of investors who instead use misspecified Black--Scholes or Dupire local volatility models. Optimal exercise boundaries are computed numerically via finite difference methods.
Stochastic volatility dynamics and return--volatility correlation are found to have a substantial impact on optimal exercise behavior across models, creating a source of model risk. As this behavior is not transmitted to exercise strategies determined by misspecified models, even if such models are fully calibrated to European option prices, calibration fails to mitigate model risk in this context. This issue persists under frequent recalibration of a misspecified model.
\end{abstract}


\section{Introduction}
\label{sec:Introduction}

The relevance of model risk in quantitative finance is well captured by George Box’s famous observation that `all models are wrong, but some are useful': financial asset price dynamics are not governed by structural laws but emerge from the interaction of market participants, institutional constraints, and information flows. As a result, probabilistic models used for pricing and risk management are inherently misspecified. The risk of losses arising from such misspecification is commonly referred to as model risk, and its importance has been widely documented in the literature on derivative pricing and hedging.

Model risk has been studied extensively in two core applications of quantitative finance: derivative valuation and hedging. In valuation, mispricing may arise either from inadequate calibration or from structural deficiencies of the pricing model, such as an insufficient representation of tail risk or an incorrect specification of the dependence between asset returns and volatility, which may lead to an underestimation of large price moves. These issues are often addressed by a worst-case approach: to quantify model risk one maximizes over a set of pricing measures representing adverse scenarios, see \cite{Cont2006}. This framework is closely related to optimal transport and to the theory of coherent and convex risk measures \citep{Artzner1999, bib:foellmer-schied-02b}.

In hedging, model risk arises through the sensitivity of hedging strategies to assumptions on the underlying dynamics, particularly in the presence of stochastic volatility or price jumps. Early contributions include the robustness analysis of hedging strategies by \cite{KarouiJeanblancShreve1998} and the uncertain volatility framework of \cite{Avellaneda1995}. Beyond model misspecification, market frictions such as limited liquidity, transaction costs, and trading constraints further affect hedging performance, see for instance \cite{bib:frey-00c} or \cite{bib:frey-polte-09}.

The above contributions focus primarily on European-style claims. In contrast, comparatively little attention has been paid to model risk for American-style derivatives and the associated optimal stopping problems.  For American options, the pricing model determines not only continuation values but also the optimal exercise strategy. Model misspecification therefore affects both the valuation of future payoffs and the timing of exercise. Even if a model is calibrated to reproduce European option prices exactly, errors in the specification of the underlying dynamics may induce suboptimal exercise policies and lead to economically significant losses. This would run contrary to the expectation that sufficient calibration to liquid option prices can ``lock down'' a misspecified model enough to eliminate most model risk. 

These issues are investigated in the present paper. We adopt the benchmark methodology for measuring model risk proposed in \cite{Hull2002}, which reflects common practices of market participants and captures key features of decision-making by option traders. We consider a continuous-time market in which the asset price follows a Heston stochastic volatility model, taken to represent the true data-generating process. The market provides quotes for European vanilla options that are consistent with this true model. Within this framework, we compare the exercise strategies of two investors holding an American put option. A sophisticated investor employs the correctly specified Heston model, calibrated to the observed option prices, to determine the optimal exercise policy for the put. An unsophisticated investor, by contrast, calibrates either a Black--Scholes model or a local volatility model in the sense of \cite{Dupire1994} to the same set of vanilla option prices and subsequently determines the exercise policy for the put within the simpler modeling framework. In this setting, we distinguish between two scenarios. First, we consider the case in which the unsophisticated investor calibrates the model only once at the initial time point. Second, we perform recalibration on a recurring basis, consistent with common market practice. For computational tractability, we restrict ourselves to the Black--Scholes model in the latter analysis.

We generate a large number of asset price trajectories under the true Heston dynamics and compare the resulting distributions of payoffs and optimal exercise times obtained by the two investors. As closed-form solutions for the exercise boundary are unavailable, optimal exercise policies are computed numerically via finite difference methods for variational inequalities, complemented by a probabilistic approach based on the Longstaff--Schwartz algorithm \citep{Longstaff2001} with randomization. 

In our numerical experiments, we first consider a baseline case in which the correlation $\rho$ between asset returns and changes in volatility is negative, as typically observed in equity and index markets, where $\rho < 0$ corresponds to a volatility skew. In line with theoretical expectations, the strategy based on the correctly specified Heston model yields the highest expected discounted payoff. The expected payoff under the Dupire local volatility model exceeds that obtained under the Black--Scholes model. Moreover, exercise tends to occur later under the Heston model than under Black--Scholes. 
We next vary the parameter $\rho$ to assess the impact of correlation on model risk. It is observed that, across all values of $\rho$, the strategy derived from the Heston model yields the highest expected payoff. Notably, for $\rho > 0$, the Black--Scholes model outperforms the Dupire local volatility model, reversing the order observed for $\rho < 0$. Moreover, for $\rho > 0$ exercise is more frequent and the payoff distribution is less skewed than for $\rho < 0$. To explain these observations we examine the correlation-sensitivity of the optimal exercise boundary in the Heston model. Interestingly, as $\rho$ increases, the exercise boundary decreases. Finally, we study calibration strategies. We find that more frequent recalibration of a misspecified model does not necessarily improve performance.

Overall, our results demonstrate that stochastic volatility and the correlation between asset returns and changes in volatility have a pronounced effect on the optimal exercise strategy for American options, thus representing a source of model risk. In particular, perfect calibration to European option prices alone is insufficient to ensure optimal exercise decisions for American-style derivatives in the presence of stochastic volatility. Ongoing recalibration to European option prices also does not bring exercise strategies determined using the wrong model to coincide with the (hypothetical) true optimal strategy, and thus both initial calibration as well as ongoing recalibration fail to mitigate the impact of model risk on exercise strategies.

\subsubsection*{Literature review} We next review the relevant literature on model risk. A first class of approaches accounts for parameter uncertainty arising from statistical estimation. In a Bayesian framework, this leads to a family of models indexed by possible parameter values, and quantities of interest are averaged across models using posterior probabilities as weights. Bayesian model averaging typically yields more stable estimates; however, it does not provide an explicit quantification of model risk. A related but non-Bayesian framework is considered in \citet{Cont2004}, where a distribution of a local volatility surface consistent with observed option prices is calibrated. Alternative methodologies employ stochastic filtering to learn unknown parameters dynamically. For example, \cite{Ekstrom2019} analyze the exercise region of American options when the underlying follows a geometric Brownian motion with unknown stochastic drift.

Worst-case or robust approaches, on the other hand, are closely related to coherent risk measures. These methods, often based on optimal transport theory, evaluate outcomes under the probability measure that generates the most adverse result within a prescribed ambiguity set. Such sets are frequently defined via the Kullback–Leibler divergence (relative entropy). For instance, \cite{Feng2018b} employ relative entropy to quantify model risk in the (re)calibration of European-type option pricing models. An alternative specification relies on the Wasserstein distance, as proposed by \citet{Feng2018}. A worst-case approach to optimal stopping problems is discussed in \cite{Nutz2015}. The literature on superreplication of European-type options, including \cite{Avellaneda1995} and \cite{Avellaneda1996}, also follows a worst-case perspective. \cite{Guo2025} establish the pricing–hedging duality for American-type options in the context of robust pricing and hedging. A common limitation of these approaches is their conservatism, which may result in hedging strategies that are prohibitively costly in practice.
In this paper, we adopt the benchmark methodology proposed in~\cite{Hull2002}. This framework assumes the existence of a true benchmark model used by a market maker to quote option prices. A second institution calibrates a misspecified model to these prices, and its pricing and hedging performance is subsequently evaluated relative to the benchmark model.
\vspace{1ex}

The remainder of the paper is structured as follows: Section \ref{sec:Models} gives a detailed description of the setup used for our analysis; Sections \ref{sec:AmericanOptionPricingOneDim} and \ref{sec:AmericanOptionPricingTwoDim} review the theory of American option pricing and discuss the numerical methods used in our analysis; in Sections \ref{sec:NumericalResultsBase} and \ref{sec:NumericalResultsExt} we discuss the results of our numerical experiments; Section~\ref{sec:Conclusion} concludes.

\section{Setup and Benchmark Methodology for Model Risk}
\label{sec:Models}

We work on a filtered probability space $(\Omega, \mathcal{F}, \mathbb{F}, \mathbb{Q})$, where $\mathbb{F} = (\mathcal{F}(t))_{t \in [0, T]}$ is the augmented filtration that represents the global information and where $\mathbb{Q}$ is a risk-neutral measure. We further assume that there are at least two assets traded continuously in a frictionless market, a riskless bank account earning a continuously compounded constant interest rate $r \geq 0$ and a stock $S$ that does not pay any dividends. We make the following assumption on the dynamics of the stock price. 

\begin{assumption}[Benchmark model]\label{ass:Heston}
The stock price $S = (S(t))_{t \in [0, T]}$ and the instantaneous variance $v = (v(t))_{t \in [0, T]}$ are the solutions to the stochastic differential equations (SDEs) 
\begin{equation} 
\label{eq:Heston}
\begin{split}
 dS(t) &= rS(t)dt + \sqrt{v(t)}S(t) dW_1(t)\,,\\
 dv(t) &= \kappa(\theta - v(t))dt + \sigma_v \sqrt{v(t)}\left(\rho dW_1(t) + \sqrt{1-\rho^2} dW_2(t)\right),
\end{split}
\end{equation}
where $W_1$ and $W_2$ are independent Brownian motions, $\kappa > 0$ is the mean-reversion rate, $\theta > 0$ is the long-term mean of $v$, $\sigma_v > 0$ is the volatility-of-volatility and $\rho \in (-1, 1)$ is the correlation between asset returns and changes in volatility. 
We assume that the Feller condition holds, that is $2\kappa \theta > \sigma_v^2$. 
\end{assumption}
Recall that $(S, v)$ is a Markov process and that the Feller condition implies that $v$ is strictly positive. We assume that the market moreover trades a set of European call options with a fairly wide range of strike and maturity pairs $(K_m, T_m), m \in \{1, \ldots, M\}$. The prices of these options are consistent with the benchmark model~\eqref{eq:Heston} and given by $C(t,S(t),v(t); K_m,T_m)$, where for generic $K$ and $T$
 \begin{equation}
 C(t, S, v; K,T) = \mathbb{E}^{\mathbb{Q}}\left[e^{-r(T-t)}(S(T) - K)^+ | S(t) = S, v(t) = v \right]\,.
\end{equation}

Following the benchmark methodology of \cite{Hull2002}, we consider two investors (banks) who seek to determine the optimal exercise strategy for an American put option, each relying on a different pricing model. The \emph{sophisticated bank} uses the Heston stochastic volatility model, calibrated to observed option prices, and therefore operates under the correctly specified benchmark model with true instantaneous variance $v(t)$. The \emph{unsophisticated bank} employs a simpler one-dimensional model with level-dependent volatility, referred to as the \emph{misspecified model}, which is also calibrated to observed option prices. We first compute the exercise strategy implied by each misspecified model. We then apply this strategy to stock price paths generated under the Heston model and compare the resulting payoff distribution with that obtained from the optimal exercise strategy derived under the correctly specified benchmark model.

We analyze two variants of the misspecified model: first, the classical Black--Scholes model \citep{Black1973} with stock price dynamics
\begin{equation}
\label{eq:BS}
d\tilde{S}(t) = r\tilde{S}(t)dt + \sigma\tilde{S}(t) dW_1(t),
\end{equation}
and second, the Dupire model \citep{Dupire1994}, where $\tilde{S}$ is the solution to the SDE 
\begin{equation}
\label{eq:Dupire}
d\tilde{S}(t) = r\tilde{S}(t)dt + \sigma(t, \tilde{S}(t))\tilde{S}(t) dW_1(t),
\end{equation}
where $\sigma(t, S)$ is a deterministic function of time and the underlying. 

The unsophisticated bank calibrates the model parameters to the observed option prices as follows. In the Black--Scholes model, $\sigma$ is set equal to the at-the-money (ATM) implied volatility, in the Dupire model, the local volatility function $\sigma(t,S)$ is calibrated to all observed option prices. Note that the calibration of the function $\sigma(t,S)$ is an ill-posed problem since only prices for finitely many strike-maturity combinations $(K_m, T_m), m \in \{1, \ldots, M\}$ are available. To deal with this issue, we use the methodology of \cite{Andersen1998}; implementation details are provided in Appendix \ref{sec:AppendixNumericsOneDim}. We consider two scenarios: in Sections~\ref{sec:NumericalResultsBase} and~\ref{sec:Correlation} the misspecified model is calibrated to the option prices from the benchmark model only at the initial date $t=0$; in Section~\ref{sec:RecalibrationBase} the unsophisticated bank recalibrates its model on a recurring basis. This case is more in line with market practice but computationally intensive, so that we restrict ourselves to the Black--Scholes model. 

The optimal exercise rule of the sophisticated bank is given by 
\begin{equation}
\label{eq:OptimalExerciseTime}
\tau^* = \inf \{t \in [0, T]: V(t, S, v) = (K - S)^+\},
\end{equation}
where the function $V(t,S,v)$ denotes the value of an American put option with strike price $K$ and maturity $T$ in the Heston model. The exercise rule employed by the unsophisticated bank, on the other hand, takes the form 
\begin{equation}
\label{eq:MisspecifiedOptimalExerciseTime2}
\tilde{\tau}^* = \inf\{t \in [0, T]: \tilde{V}(t, S) = (K - S)^+\},
\end{equation}
where $\tilde{V}$ represents the value of an American put option in the misspecified model (Black--Scholes or Dupire). We discuss the valuation of American put options in Section~\ref{sec:AmericanOptionPricingOneDim} under the Black--Scholes and Dupire models, and in Section~\ref{sec:AmericanOptionPricingTwoDim} under the Heston model.

\section{American Options in the Black--Scholes and Dupire Model}
\label{sec:AmericanOptionPricingOneDim}

In the following, we briefly review the American put pricing problem for the one-dimensional Black--Scholes and Dupire models. In both models, the value of an American put is given as
\begin{equation}
\label{eq:MisspecifiedOptimalStoppingProblem}
\tilde{V}(t, S) = \sup_{\tau \in \mathcal{T}(t)} \mathbb{E}^\mathbb{Q}\left [ e^{-r(\tau - t)} (K - \tilde{S}(\tau))^+ \mid \tilde{S}(t) = S \right ],
\end{equation}
where $\mathcal{T}(t)$ denotes the set of stopping times with respect to $\mathbb{F}$ with values in $[t, T]$. The discounted value of an American option is thus given as the Snell envelope of the discounted payoff function. Clearly, $\tilde V(t,S) \ge (K-S)^+$, and we define the exercise region $\tilde{\mathcal{E}}$ as 
$$
\tilde{\mathcal{E}} = \{(t, S) \in [0, T] \times \mathbb{R}^+: \tilde{V}(t, S) = (K - S)^+ \},
$$
and its complement, the continuation region $\tilde{\mathcal{C}}$, as
$$
\tilde{\mathcal{C}} = \{(t, S) \in [0, T] \times \mathbb{R}^+: \tilde{V}(t, S) > (K - S)^+ \}.
$$
It is well known that the optimal exercise time $\tilde \tau^*$ for the Black--Scholes or Dupire model is given by \eqref{eq:MisspecifiedOptimalExerciseTime2}, that is, by the first hitting time of the exercise region $\tilde{\mathcal{E}}$. 

For the Black--Scholes model, it has been proven that for every $t \in [0,T]$ there is exactly one value of the state variable that separates $\tilde{\mathcal{C}}$ from $\tilde{\mathcal{E}}$, that is, the set 
$\tilde{\mathcal{E}}_t:= \{S\ge 0\colon (t,S ) \in \tilde{\mathcal{E}} \}$ is of the form $[0,\tilde B(t)]$ for some threshold value $\tilde B(t) \in [0,K]$
\cite[Proposition 2.1.3, p. 3]{Jacka1991}. We next show that this holds also in the context of the Dupire model. We make use of the following result \cite[Corollary 2.6, p. 269]{Ekstrom2004}. 

\begin{proposition}[Convexity of $\tilde{V}$]
\label{cor:Convexity}
Assume that $\sigma(t, S)$ is locally Hölder$\left(\frac{1}{2}\right)$ in $S$ on $[0, T] \times \mathbb{R}^+ $ and that there exists a constant $C$ such that $\lvert \sigma(t, S)\rvert \leq C(1 + S^{-1})$ for all $t \le T$. Then, convexity and boundedness of the payoff function of an American put imply that $\tilde{V}$ is convex in $S$. 
\end{proposition}

\begin{corollary}[Unique Optimal Exercise Boundary]
\label{prop:boundaryDupire}
Suppose the assumptions of Proposition~\ref{cor:Convexity} are satisfied. Then, for the Dupire model given in~\eqref{eq:Dupire}, for every $t \in [0,T]$ there is some $\tilde B(t) \in [0,K]$ such that $\tilde{\mathcal{E}}_t =[0,\tilde{B}(t)]$.
\end{corollary}

\begin{proof}
By definition, $\tilde{V}(t,S) \geq 0$. We may thus rewrite the Dupire exercise region as 
$$\tilde{\mathcal{E}} = \{(t, S) \in [0, T] \times \mathbb{R}^+: \tilde{V}(t, S) \leq K - S \}.$$ 
For fixed $t \in [0, T]$ let $\tilde{H_t}(S) := \tilde{V}(t, S) - (K - S)$. By Proposition~\ref{cor:Convexity}, $\tilde{V}$ is convex in $S$ and, by linearity of $K-S$, $\tilde{H}_t(\cdot)$ is convex as well. 
Next, fix some $t_0$ in $[0,T]$. Then 
$\tilde{\mathcal{E}}_{t_0}$ is equal to $ \{S\geq 0\colon \tilde{H}_{t_0} (S) \leq 0 \}.$ This set is convex by the convexity of the function $\tilde H_{t_0}$ and hence an interval. Moreover, $0 \in \tilde{\mathcal{E}}_{t_0}$, which shows that $\tilde{\mathcal{E}}_{t_0}$ is of the form $ [0, \tilde{B}(t_0)]$. 
\end{proof}

In the sequel we call the mapping $t\mapsto \tilde B(t)$ the \emph{exercise boundary}. Note that by the above results we may rewrite the optimal exercise time as 
\begin{equation}
\label{eq:MisspecifiedOptimalExerciseTime}
\tilde{\tau}^* = \inf \{t \in [0, T]\colon S(t) \in \tilde{\mathcal{E}_t} \} = \inf\{t \in [0, T]\colon S(t) \leq \tilde{B}(t)\}.
\end{equation}
In our model risk analysis, we will apply this exercise rule to stock prices coming from the benchmark model. Unfortunately, there is no closed-form solution for the exercise boundary; therefore we must resort to numerical methods. 

\begin{remark}
The following qualitative properties have been established for the exercise boundary in the Black--Scholes model.
\begin{enumerate}
    \item $\tilde{B}$ is continuously differentiable on $[0, T)$ \cite[Lemma 4.1]{Myneni1992}.
    \item $\tilde{B}$ is increasing in $t \in [0, T]$ \cite[Proposition 2.2.2]{Jacka1991}.
    \item We have $\tilde{B}(T) = \lim_{t \rightarrow T} \tilde{B}(t) = K$ \cite[p. 560]{Kim1990}.
\end{enumerate}
For the Dupire model, far less is known. We are only aware of \citet[Corollary 3.2, p. 7]{DeMarco2017}, who give an expression for the exercise boundary in the limit as $t \to T$ for the case where $\sigma(t, S) \equiv \sigma(S)$. 
\end{remark}

Next, we discuss the characterization of the American put price via variational inequalities. We focus on the Dupire model; results for the Black--Scholes model correspond to the special case $\sigma(t,S) \equiv \sigma$. 
We consider the logarithmic stock price $\tilde{X}:=\log(\tilde{S})$. Using It\^o's formula, we get that 
$$
d\tilde{X}(t) = w(t, \tilde{X}(t))dt + \sqrt{v(t, \tilde{X}(t))}dW_1(t),
$$
where $w(t, \tilde{X}) = r - \frac{1}{2}\sigma(t, \exp(\tilde{X}))^2 $ and $v(t, \tilde{X}) = \sigma(t, \exp(\tilde{X}))^2$. 

For the following Theorem, assume the Dupire model is such that $\sigma(t, \exp(\tilde{X})): [0, T] \times \mathbb{R}^+ \rightarrow \mathbb{R}^+$ satisfies the usual Lipschitz and linear growth conditions. 
\begin{theorem}
\label{thm:ViscosityOneDim}
The function $\tilde{U}(t, \tilde{X}) = \tilde{V}(t, \exp(\tilde{X}) )$ is a viscosity solution to the variational inequality 
\begin{equation}
\label{eq:var-ineq-Dupire}
\max\left\{ \frac{\partial \tilde{U}}{\partial t} + \tilde{\mathcal{L}}\tilde{U} - r\tilde{U}, (K - \exp(\tilde{X}))^+ - \tilde{U}  \right\} = 0
\end{equation}
with terminal condition $\tilde{U}(T, \tilde{X}) = (K - \exp(\tilde{X}))^+$, where $\tilde{\mathcal{L}}:= w(t, \tilde{X})\frac{\partial }{\partial \tilde{X}} + \frac{1}{2} v(t, \tilde{X})\frac{\partial^2 }{\partial \tilde{X}^2}$ is the generator of $\tilde{X}$. Moreover, a comparison principle holds for the variational inequality \eqref{eq:var-ineq-Dupire}.
\end{theorem}
This is a standard result, see for instance \citet[Theorem 3.5, Chapter 3, p. 41]{Touzi2010}. 
Theorem \ref{thm:ViscosityOneDim} provides theoretical justification for approximation methods based on the variational inequality, such as the Brennan-Schwartz algorithm \citep{Brennan1977}. Under this approach, one uses standard finite difference (FD) methods to approximate the PDE part of \eqref{eq:var-ineq-Dupire}. At each time step, one checks where exercise is beneficial and adapts the value accordingly. The early exercise boundary is then implicitly determined as the maximal grid point at which early exercise is beneficial. Details of our implementation can be found in Appendix~\ref{sec:AppendixNumericsOneDim}.

\section{American Option Pricing in the Heston Model}
\label{sec:AmericanOptionPricingTwoDim}

In the following, we review the American put pricing problem for the Heston model and discuss our numerical approximation method.

\subsection{Theory}
\label{sec:AmericanOptionPricingTwoDimTheory}
In the Heston model, the price of the American option is given by
\begin{equation} 
\label{eq:price-Heston}
V(t, S, v) = \sup_{\tau \in \mathcal{T}(t)} \mathbb{E}^\mathbb{Q}\left[e^{-r(\tau - t)} (K - S(\tau))^+ \rvert S(t) = S, v(t) = v\right].
\end{equation}
Note that definition \eqref{eq:price-Heston} implicitly assumes that the martingale measure $Q$ can be determined uniquely from the observed option prices. As before, we define the \emph{exercise region} $\mathcal{E}$ as
$$
\mathcal{E} = \{(t, S, v) \in [0, T] \times \mathbb{R}^+ \times \mathbb{R}^+: V(t, S, v) = (K - S)^+\};
$$
its complement $\mathcal{C}$ is the \emph{continuation region}.
The optimal exercise time is given by $\tau^* = \inf \{t \in [0,T] \colon S(t) \in \mathcal{E}\}$. \citet[Proposition 4.1, p. 6]{Lamberton2019} prove that $V$ is convex in $S$. A similar argument as in Corollary~\ref{prop:boundaryDupire} therefore shows that there exists an optimal exercise boundary $B:[0, T] \times \mathbb{R}^+ \rightarrow \mathbb{R}^+$ that separates the continuation region from the exercise region in a way such that one may write 
$$
\mathcal{E} = \{(t, S, v) \in [0, T] \times \mathbb{R}^+ \times \mathbb{R}^+: S \leq B(t, v)\}.
$$
We may thus express the optimal exercise time in the form $\tau^* = \inf \{t \in [0, T]: S(t) \leq B(t, v(t))\}.$
\citet[Proposition 4.4, p. 9]{Lamberton2019} show that the optimal exercise boundary under the Heston model is nonincreasing and left-continuous in the volatility variable, and nondecreasing and right-continuous in time, provided that $\rho \in (-1, 1)$.

\citet[p. 416]{Touzi1999} gives a characterization of $V$ as the unique viscosity solution of the variational inequality 
\begin{equation}\label{eq:var-ineq-Heston}
\max\left\{ \frac{\partial V}{\partial {t}} + \mathcal{L}V  -rV, (K - S)^+ - V  \right\} = 0
\end{equation}
with terminal condition $V(T, S, v) = (K - S)^+$, where $\mathcal{L}:= rS\frac{\partial }{\partial S} + \kappa(\theta - v)\frac{\partial }{\partial v} +\frac{1}{2}vS^2 \frac{\partial^2 }{\partial S^2} + \frac{1}{2}\sigma_v^2 v \frac{\partial^2 }{\partial v^2} + \rho \sigma_v vS \frac{\partial^2 }{\partial S \partial v}$ is the generator of $(S, v)$.   

\subsection{Numerics}
\label{sec:AmericanOptionPricingTwoDimNumerics}

As in the one-dimensional case, we make use of the Brennan-Schwartz algorithm for the valuation of the American put option. We thus utilize the variational inequality in order to approximate the value and implicitly obtain an estimate of the optimal exercise boundary. Our approximation approach for the PDE is based on a splitting scheme of the Alternating Direction Implicit (ADI) type. Originally, ADI schemes were not developed to deal with a mixed spatial derivative as is present in the Heston PDE. \cite{Hout2010} address this and consider adaptions of numerous important ADI schemes to deal with arbitrary correlations within the Heston framework. We follow their approach for the approximation of the PDE. 
We utilize the following boundary conditions: 
\begin{itemize}
    \item $\lim_{S \rightarrow \infty^-} V(t, S, v) = 0$,
    \item $\lim_{v \rightarrow \infty^-} V(t, S,v) = K$,
    \item $V(t, 0, v) = K$.
\end{itemize}

\cite{Hout2010} follow the method-of-lines, in which one first discretizes space and then time. In order to discretize the state variables, we make use of non-uniform grids that are given below. We start off with $m_1 + 1$ equidistant points $\xi_0 < \xi_1 < \ldots < \xi_{m_1}$ given by 
$$
\xi_i = \sinh^{-1}(-K/c) + i\cdot\Delta_\xi \text{ for } i = 0, 1, \ldots, m_1,
$$
with $\Delta_\xi = \frac{1}{m_1} \left[\sinh^{-1}((s_{\max}-K)/c) - \sinh^{-1}(-K/c) \right]$ and and some constant $c > 0$. We then obtain a non-uniform mesh $s_0 = 0 < s_1 < \ldots < s_{m_1} = s_{\max}$ by setting
$$
s_i = K + c \sinh(\xi_i) \text{ for } i = 0, 1, \ldots, m_1.
$$
The parameter $c$ controls the fraction of mesh points that lie in the neighborhood of the strike price $K$. One would generally like to have enough mesh points in the neighborhood of the strike price because the initial condition has a discontinuity in its first derivative when $S = K$. We need to ensure that we are also fine enough for lower values of this state variable, as our approximation of the exercise boundary is based on this grid as well. 

We set up the non-uniform mesh for our second state variable in a similar way. It is desirable to employ a sufficiently fine mesh in the neighborhood of $0$, as the Heston PDE becomes convection-dominated for $v \approx 0$. We consider $m_2 + 1$ equidistant points given by 
$$
\zeta_j = j \cdot \Delta_\zeta \text{ for } j = 0, 1, \ldots, m_2,
$$ 
with $\Delta_{\zeta} = \frac{1}{m_2} \sinh^{-1}(v_{\max}/d)$ for some constant $d > 0$. Then the mesh $v_0 = 0 < v_1 < \ldots < v_{m_2} = v_{\max}$ is defined by
$$
v_j = d \sinh(\zeta_j) \text{ for } j = 0, 1, \ldots, m_2.
$$

We make use of downward, central and upward schemes in order to approximate the derivatives in the state variables; details are given in Appendix \ref{sec:AppendixFDSchemes}. We follow the approach outlined in \citet[p. 3018]{Hout2010} for choosing the FD scheme.
\begin{itemize}
    \item At the outflow boundary $v = 0$ we approximate $\frac{\partial V}{\partial v}$ using the upward FD scheme (given in Equation \eqref{eq:Upwind2FirstDeriv} in the Appendix). Technically, we would do so as well for $\frac{\partial^2 V}{\partial v^2}$ and use a mix of a central and upward scheme for $\frac{\partial^2 V}{\partial S \partial v}$, but all these parts vanish for $v = 0$. All remaining derivatives are approximated using a central scheme.
    \item Whenever $v \in (0, 1)$ we use central schemes only. 
    \item In the region where $v > 1$ we apply the downward scheme (given in Equation \eqref{eq:Upwind1FirstDeriv} in the Appendix) for $\frac{\partial V}{\partial v}$ to avoid spurious oscillations in the FD solution that could happen when $\sigma_v$ is close to zero. All the other derivatives are approximated using a central scheme. Within our approximations we require that $v_{\max} > 1$.
\end{itemize}

This state discretization leads to a system of stiff Ordinary Differential Equations (ODEs) 
$$V'(\hat{t}) = AV(\hat{t}) + b(\hat{t}) \text{ for } 0 \leq \hat{t} \leq T,$$ 
with initial condition $V(0) = V_0$\footnote{Following \citet{Hout2010}, we have transformed the PDE into an initial value problem by working with $\hat{t} = T-t$.}, where $A$ is a $m \times m$ matrix and $m:= (m_1 - 1)\cdot m_2$. The entries of $V(\hat{t})$ are approximations to the exact PDE solution $V(\hat{t}, S, v)$ at the spatial grid points $(S, v)$. Clearly, we need to order these grid points in a convenient way. We do so by first fixing one $v$-grid point and ranging over all possible $s$-grid values, and then moving on to the next $v$-grid point and proceeding in this manner until we reach $(s_{\max}, v_{\max})$. 
We now turn to the time discretization. For this, we fix some $\Delta_t > 0$ and define the time grid as $t_n = n \Delta_t, n = 0, 1, 2\ldots,m_3$. It is common to use the Crank-Nicolson scheme to numerically solve stiff initial value problems for systems of ODEs of the form $V'(\hat{t}) = F(\hat{t}, V(\hat{t}))$. However, because dimensionality quickly explodes in our case, such methods can become ineffective. This is where splitting schemes of the ADI-type become helpful. We thus decompose the matrix $A$ into three sub-matrices $A = A_0 + A_1 + A_2$, where 
\begin{itemize}
    \item matrix $A_0$ contains the discretization of the mixed derivative.
    \item matrix $A_1$ contains the discretization of the spatial derivatives in the $S$-direction. 
    \item matrix $A_2$ contains the discretization of the spatial derivatives in the $v$-direction.
 \end{itemize}
We split the $-rV$ term of the PDE equally between $A_1$ and $A_2$. In a similar logic, we decompose the boundary vector and define $F_j(\hat{t}, w) = A_j w + b_j(\hat{t})$ for $j = 0, 1, 2$ and $\hat{t} \in [0, T], w \in \mathbb{R}^{m}$, meaning that we split $V'(\hat{t}) = F(\hat{t}, V(\hat{t})) = F_0(\hat{t}, V(\hat{t})) + F_1(\hat{t}, V(\hat{t})) + F_2(\hat{t}, V(\hat{t}))$. We provide further details in Appendix \ref{sec:AppendixNumericsHeston}. 

\citet[p. 309-310]{Hout2010} formulate numerous splitting schemes for the underlying initial value problem; we make use of the Modified Craig-Sneyd (MCS) scheme in which $\lambda_2$ is a real parameter. The MCS scheme begins with a forward Euler step, followed by an adaption first in the $S$-direction and then in the $v$-direction. After these unidirectional corrections, the value is adjusted in the cross-direction. This is followed by a second predictor step, after which two unidirectional corrector steps are applied again. We note the exact scheme in Appendix \ref{sec:AppendixNumericsHeston}. According to \cite[p. 312]{Hout2010} the MCS scheme is unconditionally stable for two-dimensional pure diffusion equations with a mixed derivative whenever $\lambda_2 \geq \frac{1}{3}$, while in general two-dimensional equations with convections such stability results are currently lacking. The regular Craig-Sneyd (CS) scheme is unconditionally stable for two-dimensional convection-diffusion equations with a mixed derivative when its parameter is at least $\frac{1}{2}$. For the case where $\lambda_2 = \frac{1}{2}$, the MCS scheme reduces to the CS scheme, which is consequently unconditionally stable in this case. Within our computations we obtained similar results when setting $\lambda_2 = \frac{1}{2}$ as when setting $\lambda_2 = 0.4$, which is the parameter we have chosen in our final approximation. We thus have no reason to doubt the stability of our implemented MCS scheme. The MCS scheme is of order two for any given $\lambda_2$.

In a last step, we then check whether early exercise would be beneficial. For each $v$-grid point we search for the highest $S$-grid point where early exercise takes place. In doing so, we obtain an approximation of the optimal exercise boundary - specifically for each discretized $v$-value we will have a critical stock price path. Since it is computationally infeasible to make the $v$-grid fine enough for our purposes, we employ interpolation, where we make use of cubic splines with a natural end condition. 

\begin{remark}
    The numerical results in this paper are based on variational inequalities and the Brennan-Schwartz algorithm. We also experimented with simulation-based methods based on the Longstaff-Schwartz algorithm \citep{Longstaff2001} coupled with a randomization procedure first developed in \citet{Carr1998}. The results were qualitatively similar but slightly less accurate. We provide some details in Appendix \ref{sec:AppendixLSBase}. 
\end{remark}

\section{Numerical experiments}
\label{sec:NumericalResultsBase}
In this section, we present numerical results for a base case, where the correlation between asset returns and instantaneous variance is fixed to $\rho =-0.5$ and where the misspecified models (Black--Scholes and Dupire) are not recalibrated. Variations in $\rho$ and recalibration is studied in Section~\ref{sec:NumericalResultsExt}. The parameter values used are given in Table~\ref{tab:params-base}. These values are largely taken from a standard case frequently considered in the literature; see for instance \citet{Zvan1998}, \citet{Clarke1999}, \citet{Oosterlee2003}, \citet{Ikonen2008} and \citet{Vellekoop2009}. Note, however, that we take a negative value $\rho=-0.5$ for the correlation between asset returns and changes in volatility, since a negative correlation is typical for equity (index) markets. Specifically, the volatility 'skew' or 'smile', is a result of the violation of returns following a normal, and thus non-skewed, distribution. This volatility skew can be produced by allowing for a nonzero correlation between spot returns and volatility changes \citep[p. 156]{Gatheral2011}. Empirically, it is observed that equity returns are typically negatively skewed \citep[p. 81, 83]{Corrado1997}, corresponding to negative correlation under the Heston model \citep[p. 336-339]{Heston1993}. 
The parameters $K$ and $T$ denote strike price and maturity of the American put; note that the choice $K=S(0)$ implies that we consider an American put which is initially ATM. 

\begin{table}[h!]
\centering
\begin{tabular}{l | c c c c c c c c c } 
Parameter & $\kappa$ & $\theta$ & $\sigma$ & $\rho$  & $r$   & $S(0)$ & $v(0)$ & $K$ & $T$ \\ \hline
Value     & $5$      &  $0.16$  &  $0.9$   & $-0.5$ & $0.1$ & $10$   & $0.25^2$& $10$ &$1$ \\
\end{tabular}
\vspace{0.2 cm}
\caption{Parameter values for the base case.}
\label{tab:params-base}
\end{table}

\subsection{Calibration of the Black--Scholes and Dupire Model}
\label{sec:Calibration}

We compute European option prices for $100$ strike and maturity pairs $(K_m, T_m), m \in \{1, \ldots, 100\}$, where we consider four different maturities, $0.25, 0.5, 0.75, 1$, and $25$ different strikes, $7, 7.25, \ldots, 12.75, 13$. The one-dimensional models are calibrated to these prices at the initial date $t=0$. More precisely, 
the Black--Scholes model is calibrated to the call price of the strike and maturity pair $(10, 1)$, which leads to $\sigma = 0.3708353$. Further, the Dupire model is calibrated to all strike and maturity pairs. We make use of the approach developed in \cite{Andersen1998}, using a semi-implicit approach and a bicubic spline interpolation for the observed option prices. Details are given in Appendix \ref{sec:AppendixNumericsOneDim}. In Figure \ref{fig:LVSurface}, we plot the calibrated local volatility surface (the function $\sigma(t,S)$). In order to check the quality of our calibration, we computed the pricing error for the quoted option prices. We obtained a mean relative absolute error of $0.7404808\%$, which is quite small.

\begin{figure}[h!]
\centering
\includegraphics[scale=0.40]{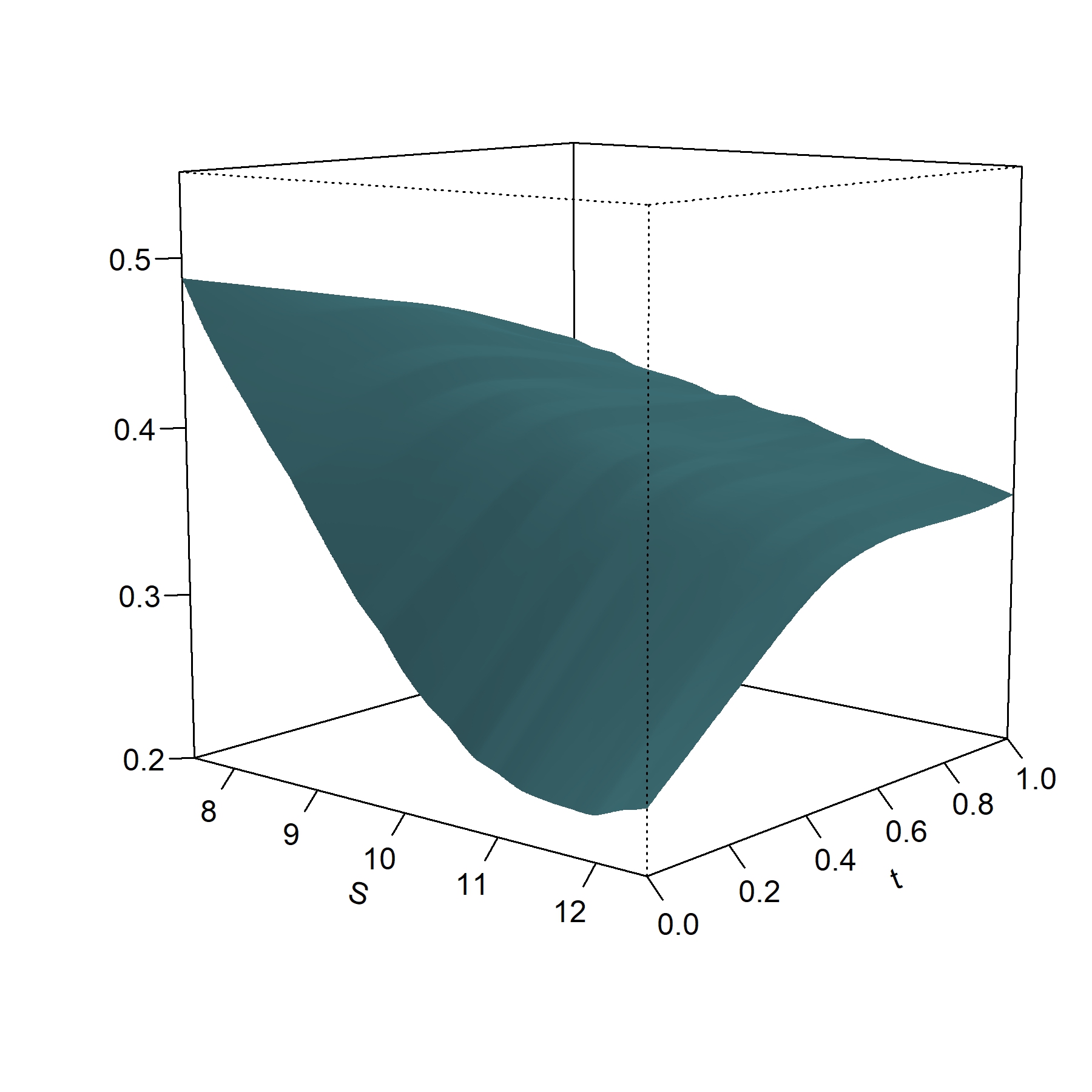}
\caption{Calibrated local volatility surface for the Dupire model.}
\label{fig:LVSurface}
\end{figure}

\subsection{Results}
\label{sec:BasecaseResults}

We generated $1$ mio. paths under the Heston model using a Milstein-discretization scheme with $300$ time-steps. The exercise boundary for the Heston, Black--Scholes and Dupire model was computed numerically via the Brennan-Schwartz algorithm; details are given in Appendix \ref{sec:AppendixNumericsOneDim} and \ref{sec:AppendixNumericsHeston}. We applied the ensuing exercise rules to the $1$ mio. Heston paths, which leads to a very precise approximation of the distributions of the discounted payoff generated by the different exercise rules. Note that for readability, in the following, the term \emph{payoff distribution} is used instead of the longer \emph{distribution of the discounted payoff}.

\subsubsection{Payoff distribution} We begin by comparing the payoff distribution generated by the three different models. Summary statistics are given in Table~\ref{tab:SummaryStats}. 

\begin{table}[h!]
\centering
    \begin{tabular}{l|cccccc}
                    & Heston & Black--Scholes  & Dupire \\
                    \hline
        Median      & 0.000   & 0.000  & 0.000 \\
        Mean (SE)   & 1.076 (0.0013)   & 1.061 (0.0012) & 1.064 (0.0013) \\
        3. Quartile & 2.334   & 2.378  & 2.383 \\
        Maximum     & 5.405   & 3.831  & 4.252 \\ \hline
    \end{tabular}
    \caption{Summary statistics of the payoff distributions of the sophisticated bank (Heston) and the unsophisticated bank (Black--Scholes/Dupire).}
    \label{tab:SummaryStats}
\end{table}
Note that the strategy based on the correctly specified Heston model yields the highest expected discounted payoff---as it should, as the optimal Heston strategy is designed to maximize the expected discounted payoff over all stopping times $\tau \in \mathcal{T}(0)$. The numerical results show that, for the given parameters, the expected payoff under the Dupire local volatility model exceeds that obtained under the Black--Scholes model. Note further that for all three strategies the median is zero, that is for more than 50\% of the simulated paths the option is not exercised at all. 

Figure~\ref{fig:QQ_FD} contains Quantile-Quantile (QQ) plots of the payoff distribution of the Heston strategy against the distribution generated by the Black--Scholes and Dupire models. We see that the medium quantiles under the Heston strategy are smaller than those of the Black--Scholes strategy, whereas the high quantiles of the Heston strategy are substantially larger than those of the Black--Scholes strategy. In particular, the payoff distribution under the Heston model does \emph{not} first-order stochastically dominate the distribution under the Black--Scholes models. The payoff distribution under the Dupire strategy is closer to the distribution of the Heston strategy. The quantiles of the Dupire strategy are mostly smaller than those of the Heston model, but, again, there is no stochastic dominance. 

\begin{figure}[h!]
\centering
\includegraphics[scale=0.50]{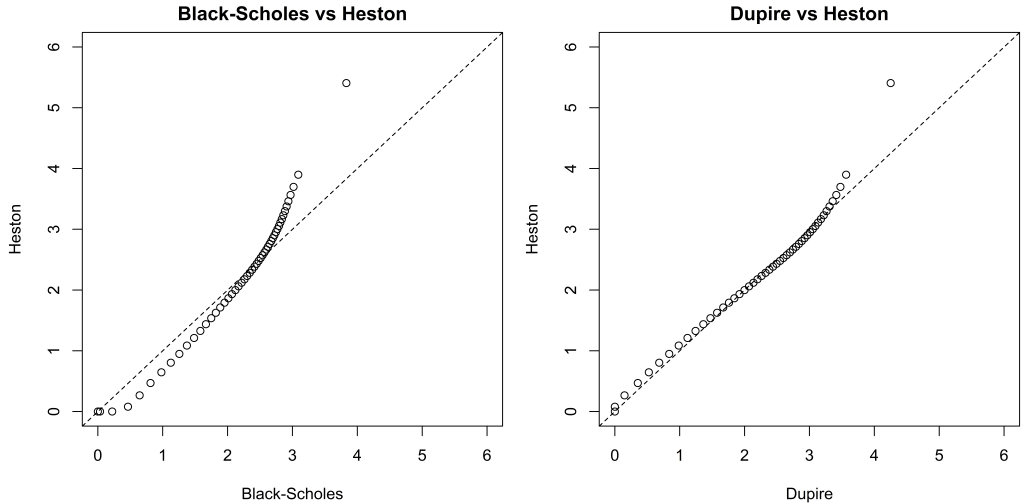}
\caption{Quantile-Quantile (QQ) plots of the payoffs of the sophisticated bank (Heston) against the unsophisticated bank (left: Black--Scholes, right: Dupire).}
\label{fig:QQ_FD}
\end{figure}

\subsubsection{Stopping times} We next analyze the stopping times.  In Figure \ref{fig:QQStopping_FD} we display QQ plots of the stopping time distribution of the Black--Scholes and Dupire strategy against the Heston strategy.  It seems that the stopping times of the Heston model stochastically dominate those of the Black--Scholes model in a first-order sense, meaning that the Heston model tends to exercise \emph{later} than the Black--Scholes model.  
On the other hand, the stopping times of the Dupire model stochastically dominate those of the Heston model in a first-order sense, i.e. the Heston model tends to exercise \textit{earlier} than the Dupire model. Further, we notice that there are more cases where the Heston model exercises compared to the Dupire model.

\begin{figure}[h!]
\centering
\includegraphics[scale=0.50]{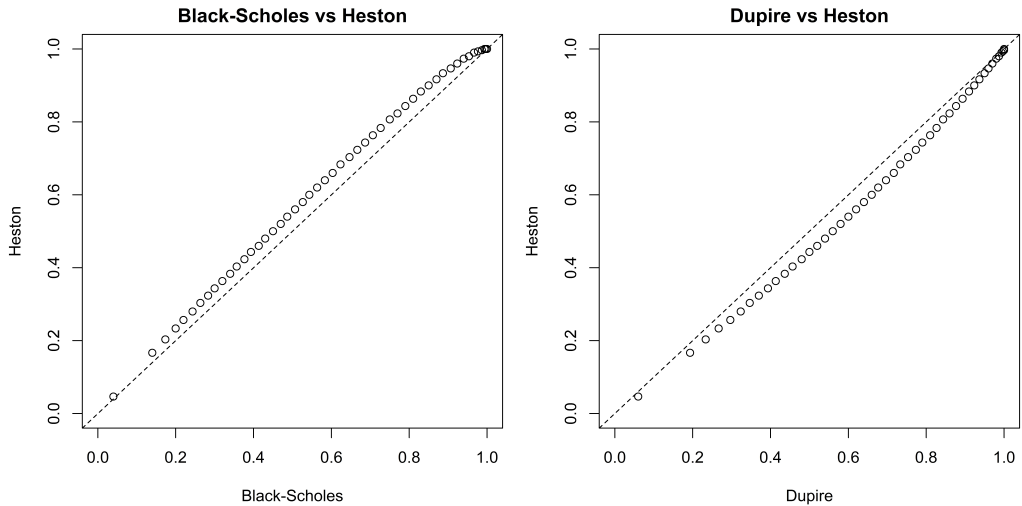}
\caption{QQ plots of stopping times of the sophisticated bank (Heston) against the unsophisticated bank (left: Black--Scholes, right: Dupire).}
\label{fig:QQStopping_FD}
\end{figure}

\subsubsection{Scatterplots} Figure \ref{fig:Scatterplot_FD} presents scatterplots of the first 10{,}000 simulated payoffs under the different exercise strategies. First, we observe that there is no clear ordering of the payoffs: there are paths where the Black--Scholes or Dupire strategy achieves a higher payoff than the Heston strategy and vice versa. In fact, there are even trajectories where one strategy exercises, leading to a positive payoff, whereas the other strategy does not exercise at all. Note that the intuition that the true model should dominate pathwise is misleading. In fact, even a random exercise strategy may outperform the optimal strategy on individual price trajectories. In optimal stopping, optimality is defined in terms of conditional expectations and not by pathwise comparison, and our results confirm that the Heston-based strategy is indeed optimal in the sense of conditional expectation.

\begin{figure}[h!]
\centering
\includegraphics[scale=0.50]{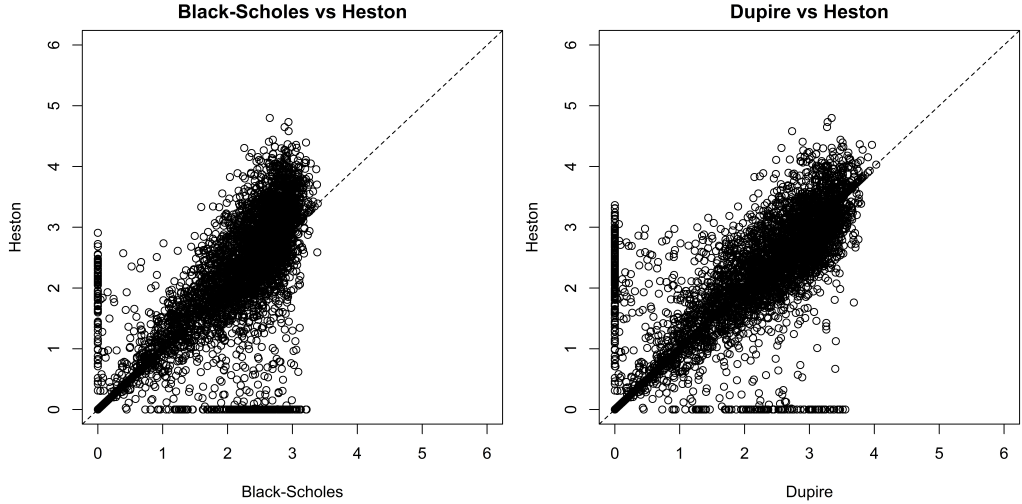}
\caption{Scatterplot of the first 10,000 payoffs of the sophisticated bank (Heston) against the unsophisticated bank (left: Black--Scholes, right: Dupire).}
\label{fig:Scatterplot_FD}
\end{figure}

\subsection{Exercise boundaries}
\label{sec:ApproximationBoundaries} 

To explain the differences among the various strategies, we examine the exercise boundaries associated with the different models. Figure \ref{fig:Boundaries} illustrates the shape of these boundaries. We observe that, in the Heston model, the exercise boundary is decreasing in the instantaneous variance level $v$. This behavior arises because the continuation value of the American put option increases with $v$, and has been proven in \citet[Proposition 4.4, p. 9]{Lamberton2019}. 
Furthermore, for large values of $v$, the exercise boundary $t \mapsto B(t,v)$ in the Heston model lies significantly below the boundary $\tilde B$ obtained under the Black--Scholes and Dupire models. This feature favors later exercise and implies that the maximal payoff under the optimal exercise strategies of the misspecified models is substantially lower than that under the Heston strategy.

\begin{figure}[h!]
\centering
\includegraphics[scale=0.5]{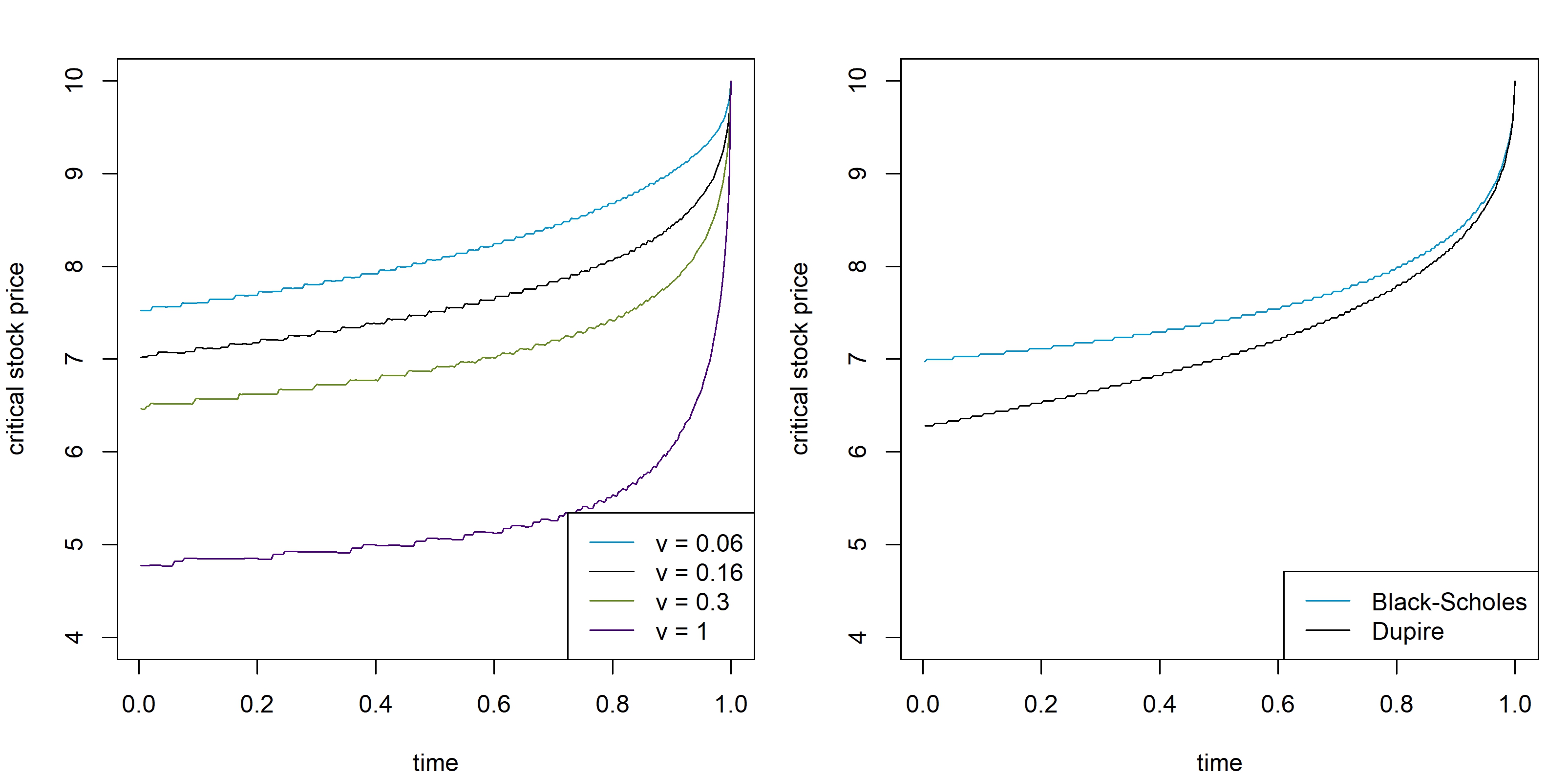}
\caption{Optimal exercise boundaries in the base case. Left: Boundary for selected volatility values under the Heston model. Right: Boundaries under the Black--Scholes and Dupire models.}
\label{fig:Boundaries}
\end{figure}

\section{Extensions}
\label{sec:NumericalResultsExt}

In this section, we discuss numerical results for two extensions of the baseline case. In Section~\ref{sec:Correlation}, we vary the correlation parameter $\rho$ between asset returns and volatility innovations, and analyze the effect on the exercise boundaries and the performance of the different strategies. In Section~\ref{sec:RecalibrationBase}, we examine the impact of periodically recalibrating the Black--Scholes model to option prices generated by the Heston model.

\subsection{Varying Correlation}
\label{sec:Correlation}

In the following, we vary the correlation parameter $\rho \in \{-0.5, 0, 0.5\}$, where $\rho =- 0.5$ corresponds to the base case. The Black--Scholes calibrated volatility is largely insensitive to this modification, with $\sigma$ remaining in the range $0.36$--$0.37$; see Table~\ref{tab:CalibratedBS_Corr}.

\begin{table}[h!]
\centering
    \begin{tabular}{lcccc}
        $\rho$           & $-0.5$    & $0$       & $0.5$      \\
                    \hline
        $\sigma$      &  0.371 & 0.370  & 0.366   \\
    \end{tabular}
    \caption{Calibrated volatility under the Black--Scholes model when varying correlation.}
\label{tab:CalibratedBS_Corr}
\end{table}

The local volatility surfaces for the Dupire model are shown in Figure \ref{fig:LVSurface_Corr}.
Note that the calibrated local volatility surfaces are sensitive to correlation, as the value of $\rho$ affects the form of the implied volatility surface generated by the Heston model: under negative correlation, local volatility is decreasing in $S$, ensuring that  the Dupire model generates a volatility skew; for positive correlation local volatility is increasing in $S$ so that the Dupire model generates a reverse skew; for $\rho=0$ the surface is comparatively flat and the Dupire model generates a volatility smile. 

\begin{figure}[htbp]
    \centering
    \begin{minipage}[b]{0.3\textwidth}
        \centering
        \includegraphics[width=\textwidth]{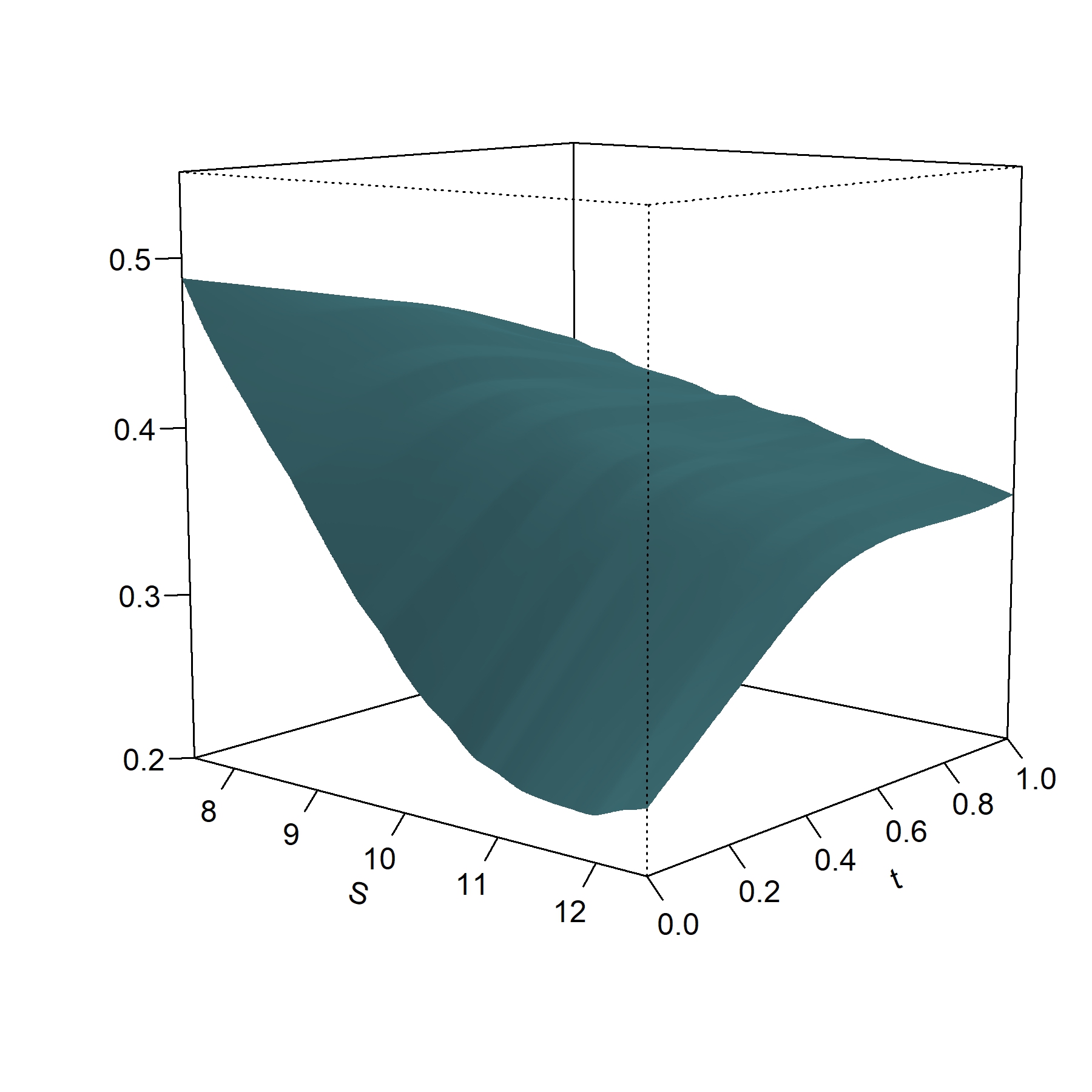}
    \end{minipage}
    \hfill
    \begin{minipage}[b]{0.3\textwidth}
        \centering
        \includegraphics[width=\textwidth]{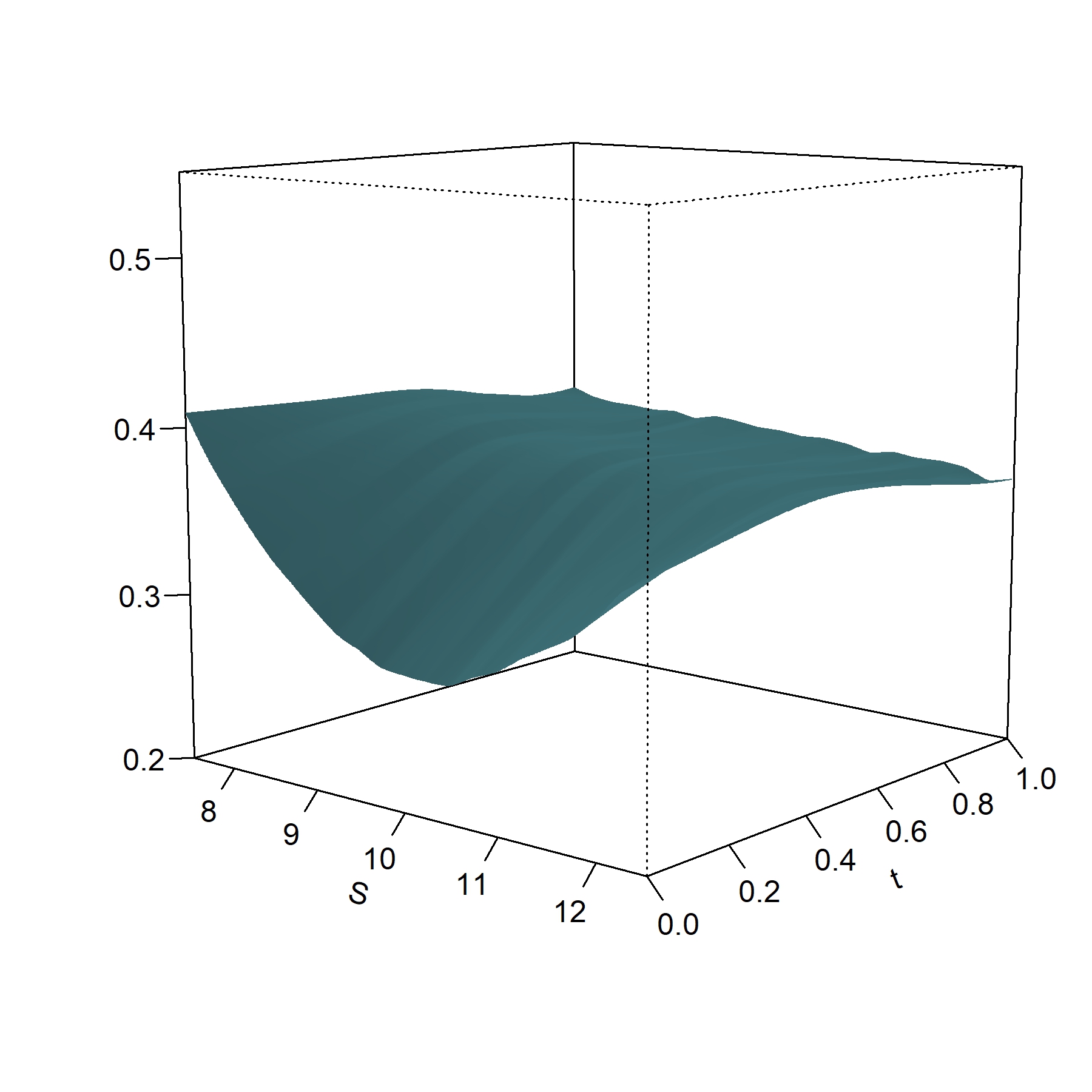}
    \end{minipage}
    \hfill
    \begin{minipage}[b]{0.3\textwidth}
        \centering
        \includegraphics[width=\textwidth]{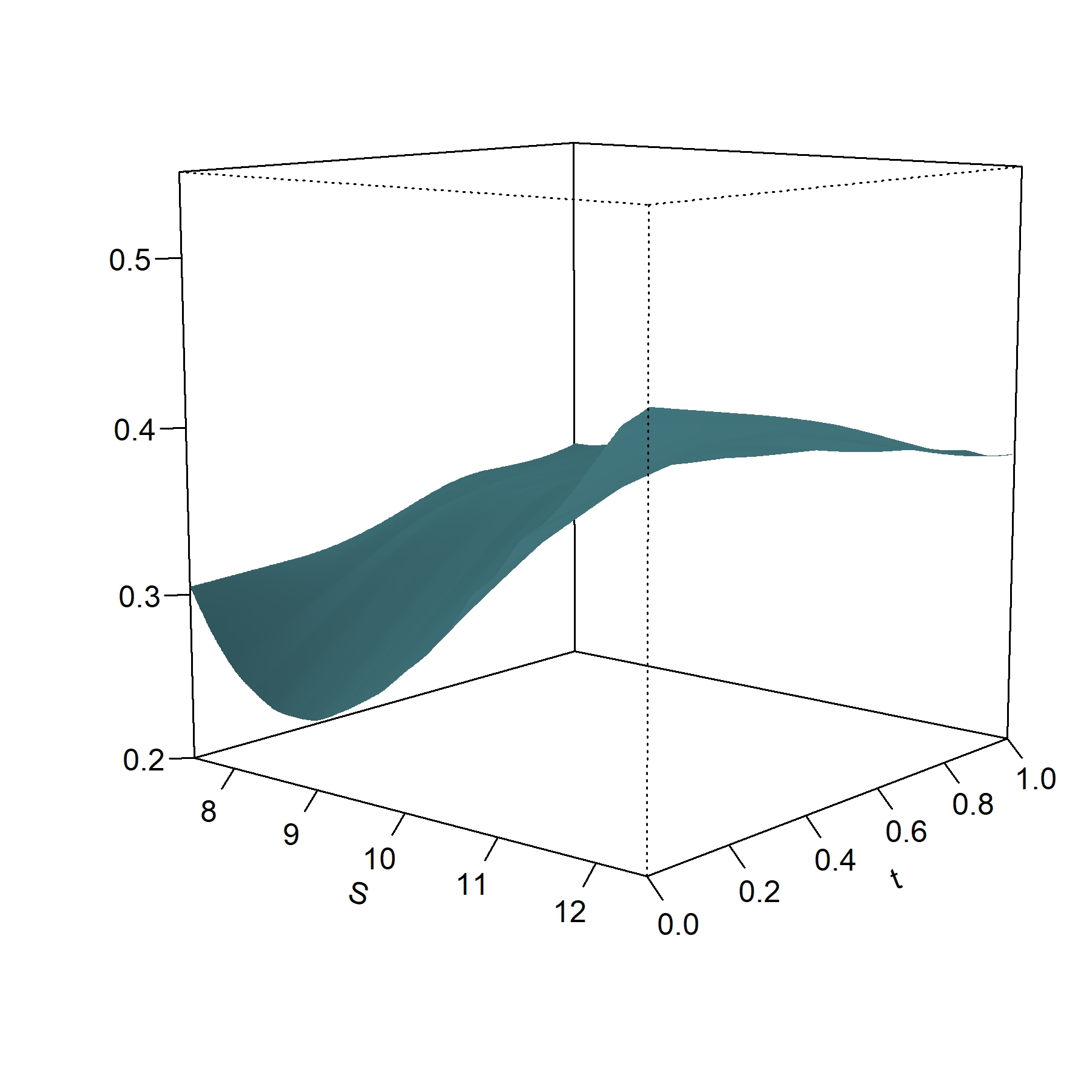}
    \end{minipage}
    \hfill
    \caption{Calibrated volatility surface under the Dupire model when varying correlation $\rho \in \{-0.5, 0, 0.5\}$ going from left to right.}
    \label{fig:LVSurface_Corr}
\end{figure}

\subsubsection{Correlation and exercise boundaries} Consider now the impact of changes in $\rho$ on the optimal exercise boundaries. The boundaries for the one-dimensional models are depicted in Figure \ref{fig:OneDimBoundaries_Corr}. We notice that the exercise boundaries for the Black--Scholes model vary very little with changes in $\rho$, in line with the small change in implied volatility. In the Dupire model, the exercise boundary is \emph{increasing} in $\rho$, that is, ceteris paribus, a higher correlation implies earlier exercise. This can be explained by the form of the local volatility surface. For positive $\rho$, the local volatility surface attains relatively low values when the American put is in the money, i.e., $S < K$. This leads to a low continuation value of the option and hence to a comparatively high exercise boundary. 
When $\rho$ is negative, on the other hand, the local volatility surface attains relatively high values for $S < K$, inducing a higher continuation value and a lower exercise boundary.

Interestingly, in the Heston model the effect of $\rho$ on the exercise boundary is reversed. As shown in~Figure \ref{fig:HestonBoundaries_Corr}, for fixed variance level $v$ the exercise boundary is \emph{decreasing} in $\rho$. To understand this observation, note that the continuation value has similar qualitative properties as the price of an in-the-money (ITM) European put option, even if the two are not identical. It has been shown for the Heston model that the value of ITM put options is increasing in $\rho$; see, for example, \citet[Theorem 4.1, p.~8]{Ould2013}. This suggests that the continuation value is increasing in $\rho$, which implies that the exercise boundary is, in return, decreasing in the correlation parameter.\footnote{Further, it has been shown that the price of out-the-money (OTM) put options is decreasing in $\rho$, but since an American put is never exercised when $S>K$, the exercise boundary is related to put options that are ITM.} 
Figure~\ref{fig:HestonBoundaries_Corr} moreover shows that for fixed $\rho$ the boundary is decreasing in $v$. This effect was discussed already in Section~\ref{sec:ApproximationBoundaries}.

Summarizing, the above observations show that the Heston model with two state variables is able to separate the effects of the volatility level from those of the correlation between stock returns and changes in volatility. In contrast, within the Dupire framework, both effects must be captured implicitly through the local volatility surface. This has the interesting consequence that calibrating the local volatility model to a full (European) option price surface generated by the Heston model fails to correctly transmit the impact of the correlation between asset returns and changes in volatility on optimal exercise strategies. In other words, calibration to the entire option price surface does not mitigate this particular source of model risk.

\begin{figure}[h!]
\centering
\includegraphics[scale=0.5]{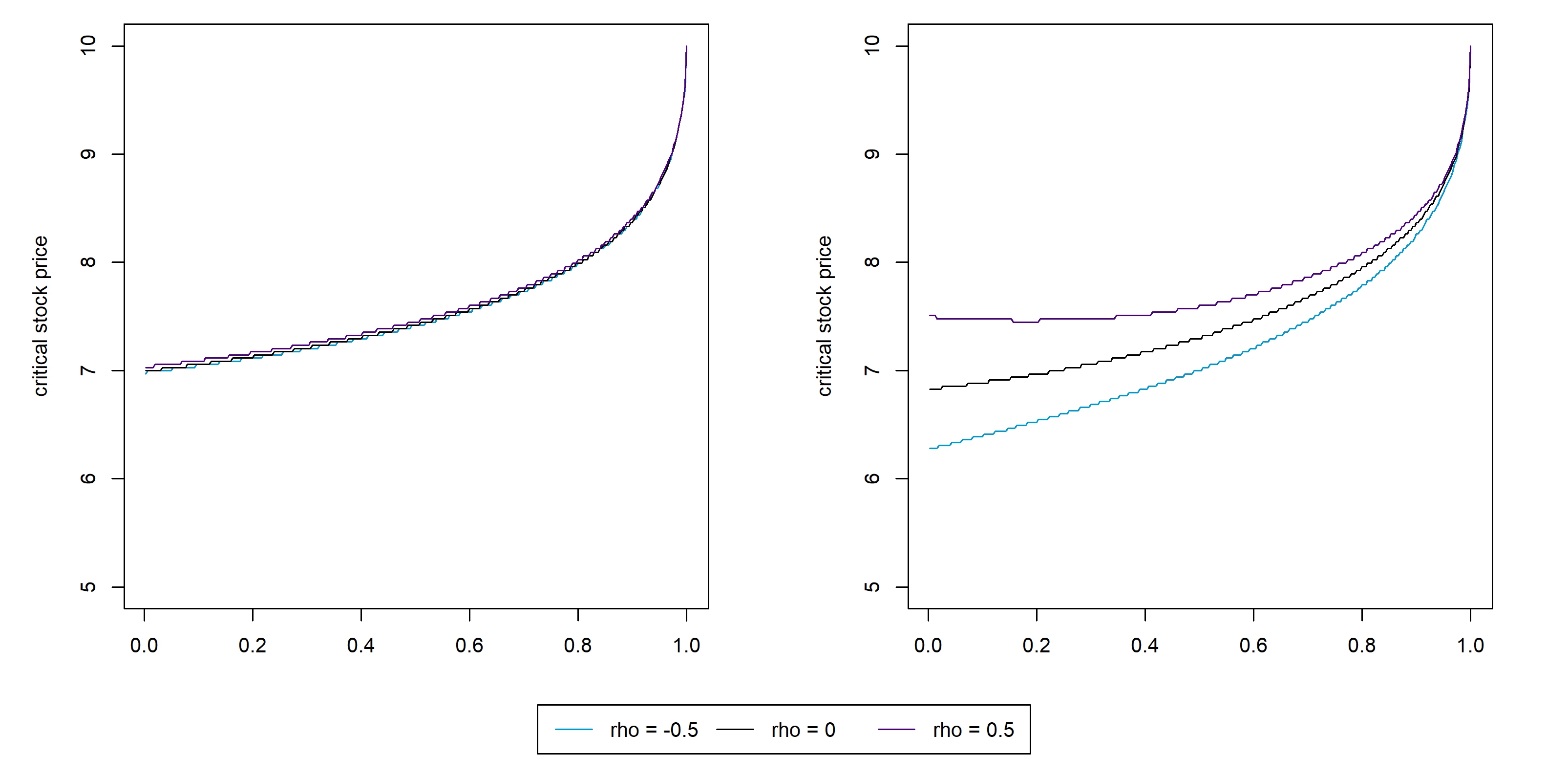}
\caption{Approximate optimal exercise boundaries when varying correlation. Left: approximate boundary under the Black--Scholes model. Right: approximate boundaries under the Dupire model.}
\label{fig:OneDimBoundaries_Corr}
\end{figure}

\begin{figure}[h!]
\centering
\includegraphics[scale=0.75]{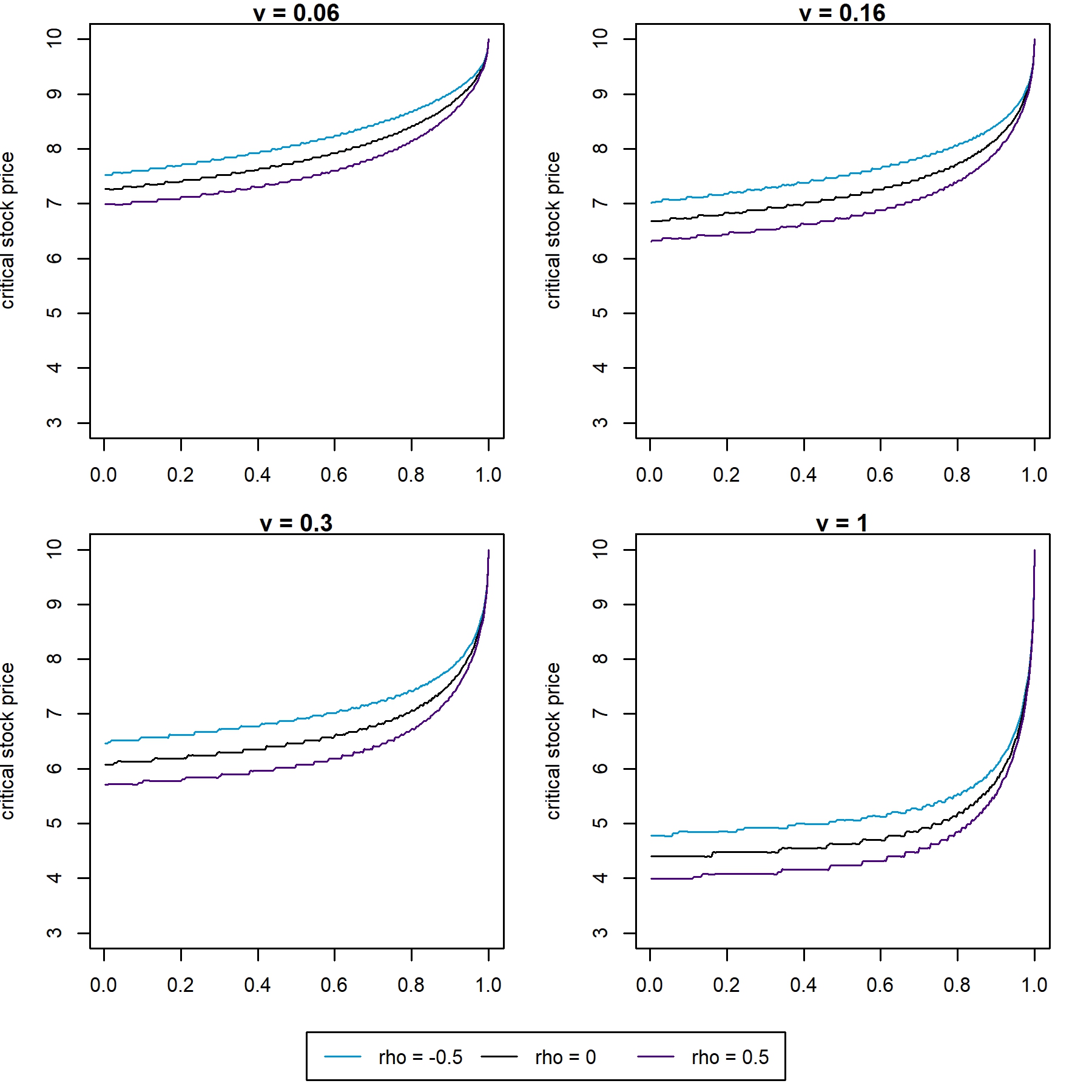}
\caption{Approximate optimal exercise boundaries for selected volatility values under the Heston model when varying correlation.}
\label{fig:HestonBoundaries_Corr}
\end{figure}

\subsubsection{Payoff distribution and stopping times} We report the expected discounted payoffs for the three models under the correlation values $\rho \in \{-0.5, 0, 0.5\}$ in Table~\ref{tab:ExpectedDiscPayoff_Corr}. Consistent with theoretical predictions, the Heston-based strategy yields the highest expected payoff for all values of $\rho$. Notably, for $\rho \leq 0$, the Dupire strategy produces a higher average payoff than the Black--Scholes strategy, whereas for $\rho = 0.5$, the Black--Scholes strategy slightly outperforms the Dupire approach.\footnote{One could speculate that this is a reflection of the issue discussed in the previous paragraph, i.e., the failure of the calibrated local volatility model to correctly capture the impact of the Heston $\rho$ on the optimal exercise strategy.}

Next we discuss the impact of $\rho$ on qualitative properties of the payoff distribution. We give the median discounted payoffs in Table \ref{tab:MedianDiscPayoff_Corr}. 
As $\rho$ increases, the median increases substantially and the difference between the mean and the median of the payoff distribution becomes much smaller for all strategies, indicating that payoff distributions become less skewed as correlation increases. Moreover, the fact that for $\rho=0.5$ the median is positive implies that in this case the option is exercised in more than 50\% of the cases under all three strategies. 

\begin{table}[h!]
\centering
    \begin{tabular}{l|cccccc}
        $\rho$  & Heston   & Black--Scholes & Dupire \\ \hline
        -0.5    & 1.076    & 1.061         & 1.064  \\
        0       & 1.061    & 1.044         & 1.046  \\
        0.5     & 1.035    & 1.023         & 1.020  \\ \hline
    \end{tabular}
    \caption{Expected discounted payoff of the sophisticated bank (Heston) and the unsophisticated bank (Black--Scholes/Dupire) when varying correlation.}
    \label{tab:ExpectedDiscPayoff_Corr}
\end{table}

\begin{table}[h!]
\centering
    \begin{tabular}{l|cccccc}
        $\rho$  & Heston   & Black--Scholes & Dupire \\ \hline
        -0.5    & 0.000    & 0.000         & 0.000  \\
        0       & 0.071    & 0.226         & 0.127  \\
        0.5     & 0.523    & 0.569         & 0.689  \\\hline
    \end{tabular}
    \caption{Median of the discounted payoff of the sophisticated bank (Heston) and the unsophisticated bank (Black--Scholes/Dupire) when varying correlation.}
    \label{tab:MedianDiscPayoff_Corr}
\end{table}

The impact of changes in the correlation parameter $\rho$ on both the optimal exercise time and the payoff distribution is driven by two interacting mechanisms. First, $\rho$ affects the location of the optimal exercise boundary. Second, it modifies the joint distribution of the underlying asset price paths in the Heston:
for $\rho < 0$, the return distribution exhibits a heavier lower tail, implying a higher probability of paths along which the stock price attains very low levels. Since the discounted stock price process is a martingale, this increased downside risk is offset by a larger mass of paths that remain close to or slightly above the initial value $S(0)=K$. Conversely, for $\rho > 0$, this asymmetry is reversed, and a larger proportion of paths evolves slightly below the strike level $K$. 
The observed impact of changes in correlation thus reflect the combined effect of boundary shifts and changes in the stochastic dynamics of the underlying variables.
\begin{itemize}
    \item In the Black--Scholes model, the optimal exercise boundary is essentially insensitive to changes in the correlation parameter. Hence, differences in the exercise strategy are driven almost entirely by changes in the distribution of the underlying Heston price paths. For small values of $t$, fewer trajectories cross the exercise boundary when $\rho > 0$ than when $\rho < 0$. In contrast, for $t$ close to $T$, exercise occurs more frequently for $\rho > 0$. As $t \to T$, the exercise boundary converges to $K$, so that even a small downward movement in the asset price---which is more likely under $\rho > 0$ than under $\rho < 0$---is sufficient to trigger exercise. This shift in exercise behavior is illustrated in Figure \ref{fig:QQ_StoppingTimesInternal_Corr}. Since the exercise boundary is increasing over time in the Black--Scholes model, later exercise moreover results in lower payoffs and a less skewed payoff distribution.
    \item In the Dupire model, the exercise boundary corresponding to $\rho = 0.5$ lies above that for $\rho = -0.5$, in particular for small $t$. Ceteris paribus, this favors early exercise.  Combined with the increased frequency of exercise for $t$ close to $T$ under positive correlation, this implies that exercise is more likely  when $\rho > 0$. Because a higher exercise boundary mechanically reduces the maximal attainable payoff, for $\rho>0$ realized payoffs are typically lower than in the base case.
    \item In the Heston model, as the correlation parameter changes from negative to positive, the exercise boundary shifts downward and the distribution of asset prices is altered. For small $t$, these effects result in less frequent exercise when $\rho > 0$. Near maturity, however, exercise becomes more frequent under positive correlation due to the behavior of the Heston price paths.  This pattern is confirmed in Figure \ref{fig:QQ_StoppingTimesInternal_Corr}: the lower (upper) quantiles of the stopping time distribution are larger (smaller) for $\rho > 0$ than for $\rho < 0$.
\end{itemize}
Overall, these observations explain why, for positive correlation, a larger fraction of paths leads to early exercise, while the corresponding payoffs tend to be lower on average, producing a less skewed payoff distribution. 

\begin{figure}[h!]
\centering
\includegraphics[scale=0.6]{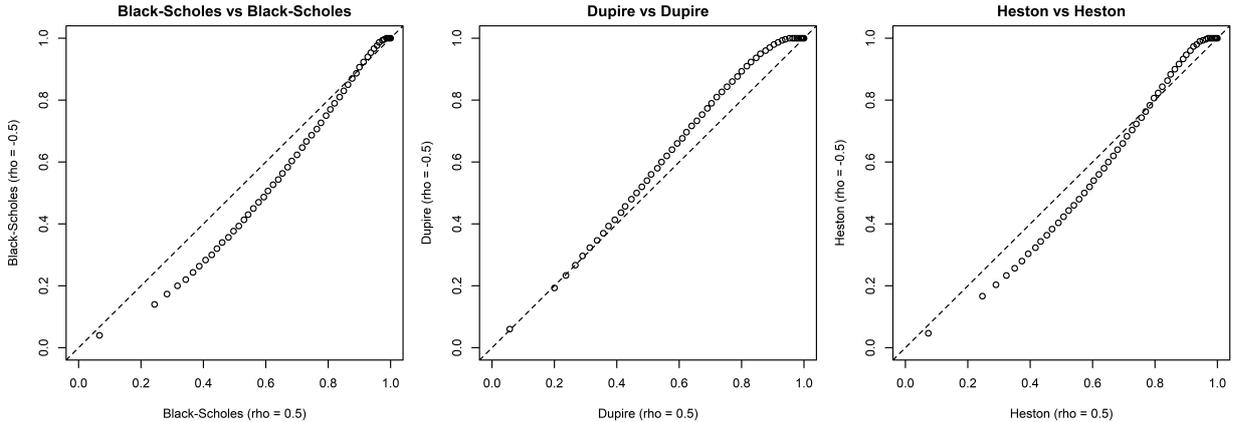}
\caption{QQ plots of the stopping times obtained under negative correlation against those obtained under positive correlation, within each model.}
\label{fig:QQ_StoppingTimesInternal_Corr}
\end{figure}

\begin{figure}
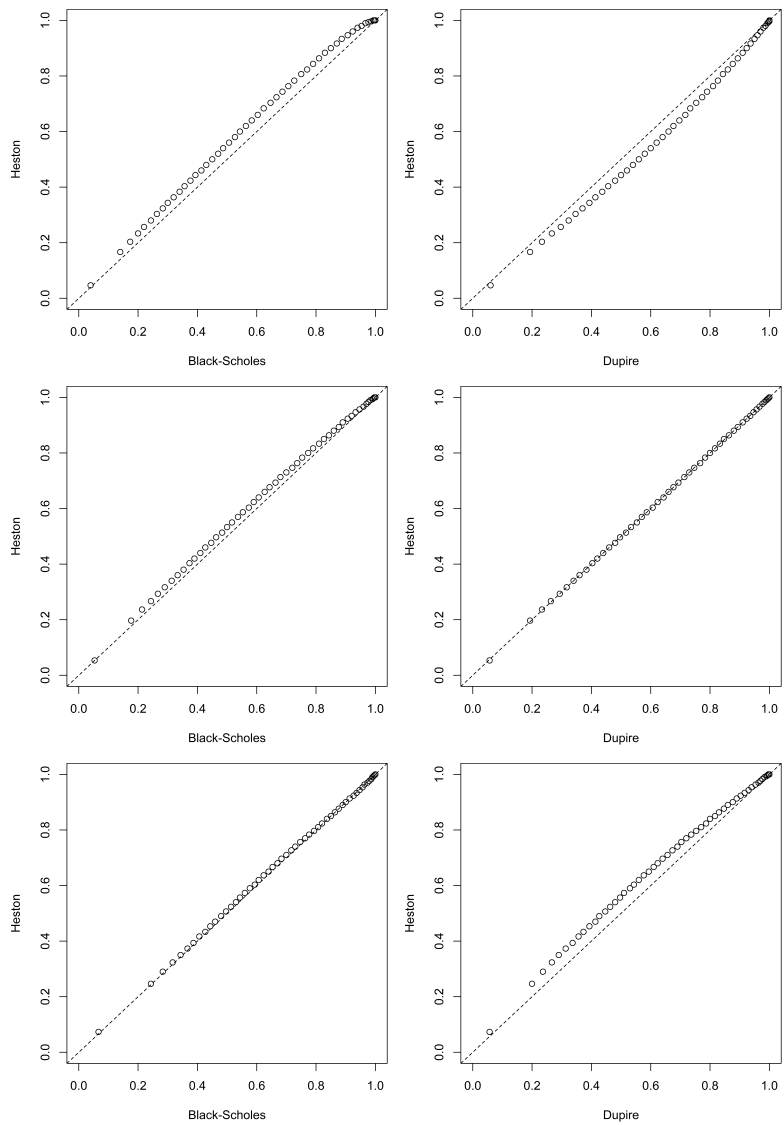

    \centering
    \includegraphics[height=0.2\textheight,trim=0cm 0cm 0cm 1cm, clip]{Correlation_Figures/QQ_StoppingTimes-05.jpeg}\par\medskip
    \includegraphics[height=0.2\textheight,trim=0cm 0cm 0cm 1cm, clip]{Correlation_Figures/QQ_StoppingTimes0.jpeg}\par\medskip
    \includegraphics[height=0.2\textheight,trim=0cm 0cm 0cm 1cm, clip]{Correlation_Figures/QQ_StoppingTimes05.jpeg}\par\medskip
    \caption{QQ plots of the stopping times of the sophisticated bank (Heston) against the unsophisticated bank (left: Black--Scholes, right: Dupire) when varying correlation $\rho \in \{-0.5, 0, 0.5\}$ (top to bottom).}
    \label{fig:QQStopping_FD_Corr}
\end{figure}

\subsection{Black--Scholes Model Recalibration at a 7-Day Frequency}
\label{sec:RecalibrationBase}

Model recalibration is standard practice in derivatives markets. In this section, we therefore analyze the performance of the Black--Scholes strategy under the assumption that the model is recalibrated at a $7$-day frequency to option prices generated by the true Heston model.
At each calibration date, the Black–Scholes model is fitted to the price of a European put with the same strike as the American put under consideration. Since the objective of the unsophisticated bank in this study is solely to determine the optimal exercise policy for a given American put, this calibration approach is internally consistent. In practice, however, financial institutions typically manage portfolios containing multiple options written on the same underlying asset. Calibration is therefore commonly based on options that are ATM at the calibration date.

The computational cost of this simulation experiment is substantial. We therefore restrict the analysis to $100,000$ simulated paths from the Heston model. As roughly $50$ recalibrations are required per path, this results in a total of roughly $5$ million calibrations and applications of Brennan-Schwartz. 

\subsubsection{Results}
\label{se:RecalibrationResults}

We now analyze the performance of the exercise strategy derived from the recalibrated Black--Scholes model. Summary statistics for the Heston strategy, as well as for the Black--Scholes strategy with and without recalibration, are reported in Table~\ref{tab:SummaryStatsRecalibrated}. The expected payoff generated by the recalibrated Black--Scholes strategy is significantly lower than that obtained under the Heston strategy. Moreover, it is slightly lower than the expected payoff of the Black--Scholes strategy without recalibration.\footnote{This may be due to simulation noise.} These results indicate that repeated recalibration does not, on average, lead to improved performance. 
This finding is noteworthy, as it provides a complementary perspective to the results in \cite{Cont-2025}, who argues that continuously recalibrating an auxiliary pricing model is an effective recipe for the dynamic hedging of European options under model uncertainty. Our results indicate that such benefits do not necessarily carry over to the context of optimal exercise for American options.

\begin{table}[h]
\centering
           \begin{tabular}{lccccc}
                   & & Heston          && \multicolumn{2}{c}{Black--Scholes}\\ 
                    &&                 && with recalibration& without recalibration \\
                    \hline
        Median      && 0.000           && 0.132           &  0.000  \\
        Mean (SE)   && 1.078 (0.0041)  && 1.057 (0.0037)  &  1.063 (0.0038)  \\
        3. Quartile && 2.337           && 2.293           &  2.377 \\
        Maximum     && 5.405           && 3.862           &  3.613\\  \hline
    \end{tabular}
    \vspace{0.1cm}
    \caption{Summary statistics of the discounted payoff distributions of the sophisticated bank (Heston) and the unsophisticated bank (Black--Scholes with and without recalibration).}
    \label{tab:SummaryStatsRecalibrated}
\end{table}

A more detailed analysis shows that the recalibrated Black--Scholes model exercises marginally more frequently than the single-calibration specification. In particular, intermediate and late exercise times are shifted slightly forward. This behavior can be explained by the properties of the implied volatility distribution at the recalibration dates.
As illustrated in Figure~\ref{fig:RecalibrationDensitiesIVs}, the distribution of implied volatilities remains centered around the initial implied volatility. As maturity approaches, however, the mean of the distribution falls below the initial implied volatility. In this case, the exercise boundary shifts upward, thereby favoring earlier exercise. This mechanism accounts for the reduction in skewness relative to the single-calibration model. Nevertheless, these shifts do not generate systematic performance improvements. At higher payoff quantiles, the single-calibration model often attains slightly higher outcomes, as shown in Figure~\ref{fig:Scatterplot_RecalBSvsSingleCal}. The persistent proximity of the implied volatility distribution to its initial level provides a plausible explanation for the absence of aggregate performance gains under the adopted recalibration procedure. Although recalibration modifies exercise decisions, it does not materially affect average results and therefore does not compensate for model misspecification.
We attribute this finding primarily to systematic discrepancies between the instantaneous volatility level $\sqrt{v(t_n)}$ at calibration dates $t_n = 1, \ldots, 43$ and the corresponding implied volatility $\sigma_n$. When the variance process $v$ is high, $\sigma_n$ tends to underestimate $\sqrt{v(t_n)}$, whereas in low-volatility regimes it typically overestimates it. This pattern is consistent with the interpretation of Black--Scholes implied volatility as a reflection of expectations of future volatility rather than contemporaneous spot volatility. When realized volatility reaches extreme levels, implied volatilities provide a stabilized estimate, particularly for long times to maturity. 


\begin{figure}[h!]
\centering
\includegraphics[scale=0.5]{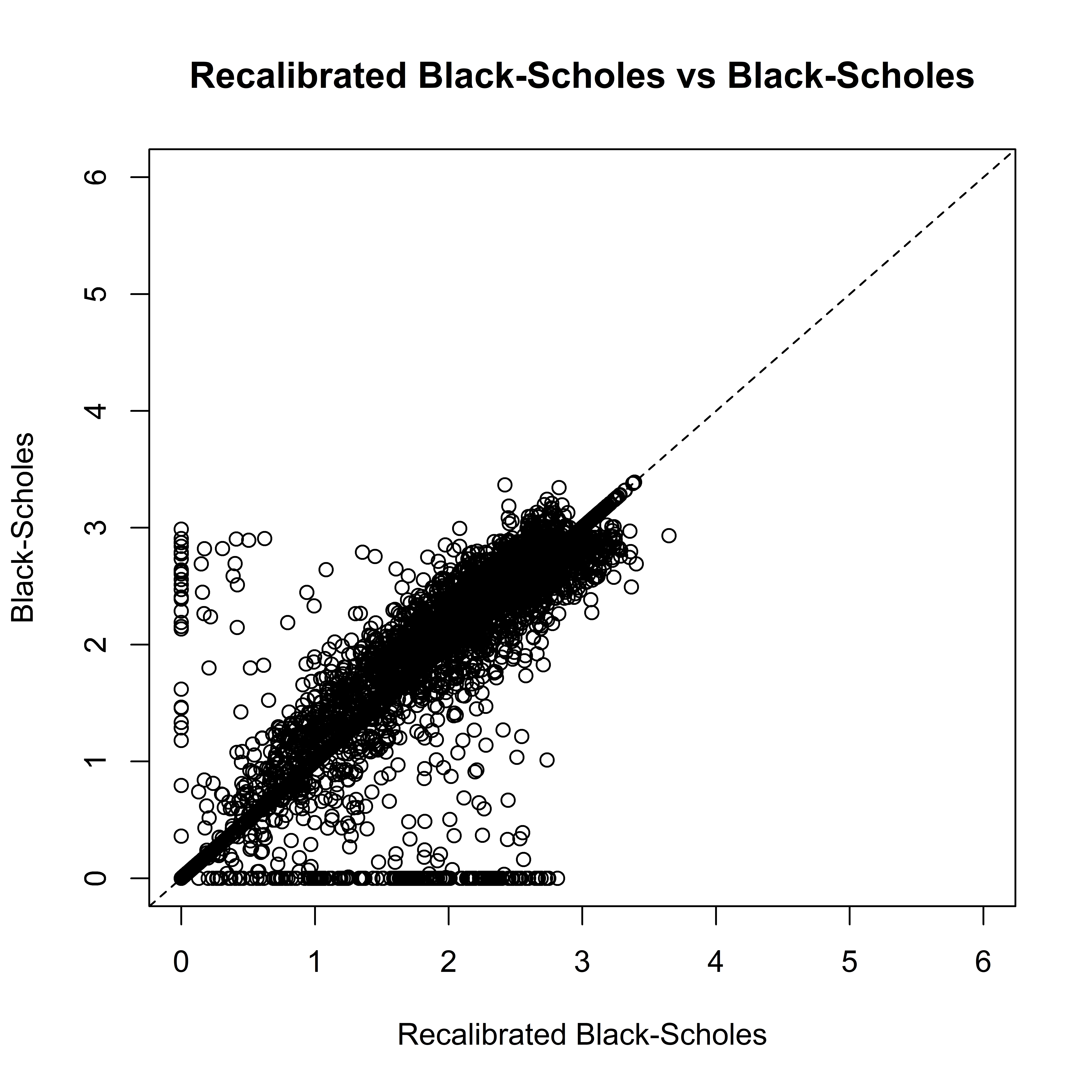}
\caption{Scatterplots of the first 10,000 payoffs of the Black-Scholes model based on initial calibration against those of the recalibrated Black-Scholes model.}
\label{fig:Scatterplot_RecalBSvsSingleCal}
\end{figure}

In summary, the results indicate that attempting to track \textit{current volatility} using implied volatilities---that is, market expectations of future volatility---does not improve performance. In particular, in our context recalibration seems to be  insufficient to offset the effects of model misspecification.   This finding is practically relevant, as the recalibration scheme closely mirrors standard market practice. 
Nevertheless, alternative recalibration methodologies may produce different results and therefore represent a promising direction for future research.

\begin{figure}[h!]
\centering
\includegraphics[scale=0.6]{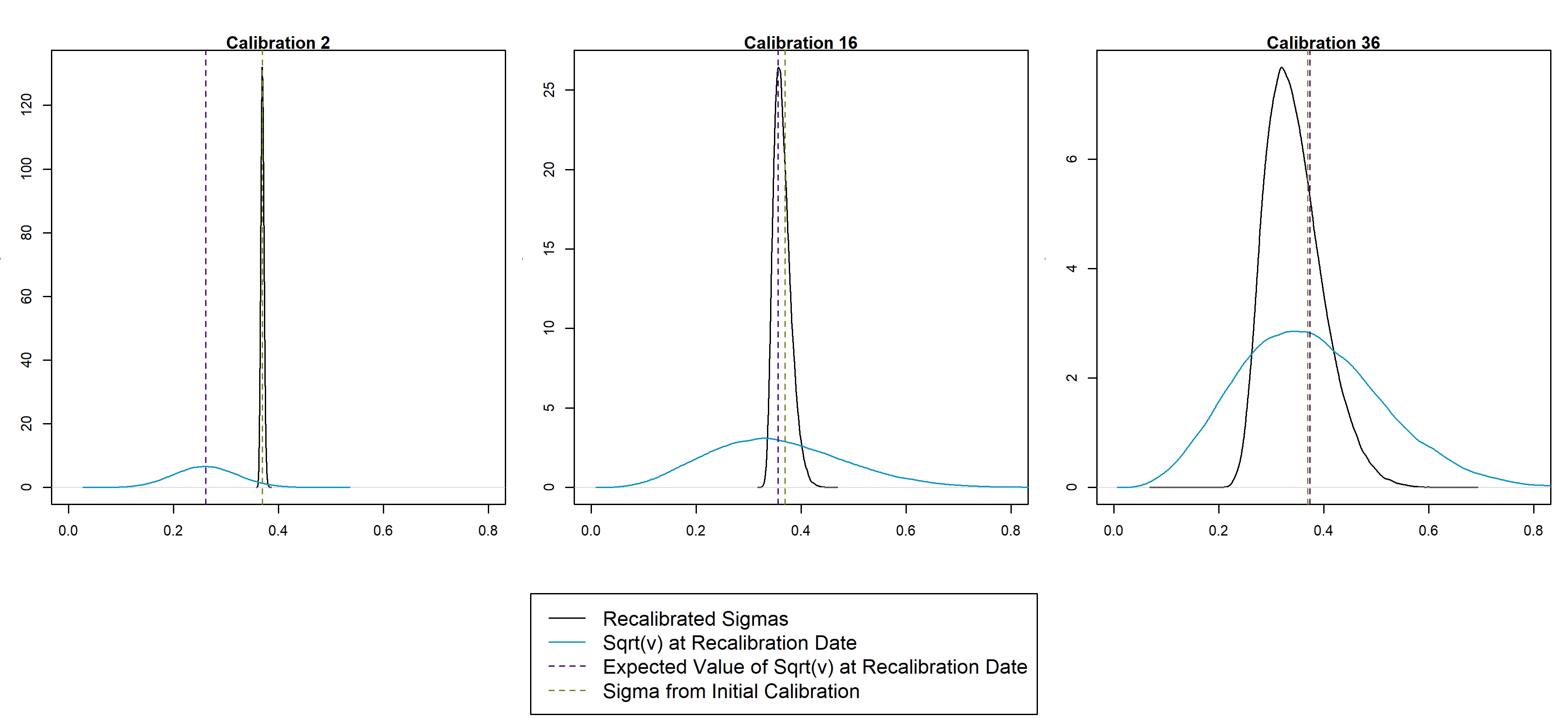}
\caption{Densities of Recalibrated Black-Scholes Implied Volatilities $\sigma_n$ and $\sqrt{v(t_n)}$ at Calibration Dates $t_n \in \{0.02, 0.347, 0.813\}$.}
\label{fig:RecalibrationDensitiesIVs}
\end{figure}

\section{Conclusion}
\label{sec:Conclusion} 
In this paper, we have explored the impact of model risk on optimal exercise strategies of American options. To this end, we considered the performance of exercise strategies derived from a ``wrong'' model compared to the truly optimal strategy matching the (hypothetical) ``true'' data generating process. Specifically, we used strategies based on the Black--Scholes and Dupire models within a benchmark framework that treats the Heston model as the data-generating process. In this framework, we considered the impact of stochastic volatility, specifically the level of volatility and the correlation between asset returns and changes in volatility, on the exercise strategies based on the misspecified and ``true'' models.

We first study a baseline setting with negative correlation between asset returns and changes in volatility. Unsurprisingly, the Heston-based strategy is optimal with respect to its own objective criterion. Nevertheless, we document numerous scenarios in which misspecified one-dimensional models outperform the benchmark ex post, purely due to sampling variation. Importantly, we do not observe stronger notions of optimality: in particular, there is no evidence of first-order stochastic dominance of the Heston-based strategy over either the Black--Scholes or the Dupire strategy. In short, strategies that are optimal in expectation are on some paths outperformed by suboptimal strategies.

Varying the correlation parameter of the benchmark Heston model provides additional insight into the sensitivity of optimal exercise boundaries to correlation. While the exercise boundaries in the one-dimensional models increase with correlation, the Heston boundary decreases as the correlation parameter rises. Since the correlation parameter $\rho$ governs the skewness of equity returns in the Heston framework, we highlight that the impact of $\rho$ on the early exercise decision simultaneously depends on the change in the exercise boundaries and the change in the market. We find that the skewness of spot returns directly impacts the skewness of the resulting discounted payoff distribution. In consequence, calibrating the Dupire model to a full European option price surface generated by the Heston model fails to correctly transmit the impact of the correlation between asset returns and changes in volatility on optimal exercise strategies. In other words, calibration to the entire option price surface does not mitigate this particular source of model risk; nor does calibrating a single implied volatility when Black–Scholes governs the exercise strategy.

An analysis of the impact of dynamic recalibration on the Black–Scholes exercise rule reveals that repeated recalibration yields negligible benefits compared to leaving the initially calibrated implied volatility unchanged. Thus, ongoing recalibration to European option prices also does not sufficiently lock the behavior of the wrong model to the (hypothetical) true dynamics, again failing to mitigate the impact of model risk on exercise strategies. This may be of particular interest to practitioners, as conventional wisdom seems to be that sufficient calibration to liquid option prices can ''lock down'' a misspecified model enough to eliminate most model risk. This may be true in some cases, for example \cite{Cont-2025} argues that continuously recalibrating an auxiliary pricing model is an effective recipe for the dynamic hedging of European options under model uncertainty, but this does not readily apply to optimal exercise strategies for American options. 

Given that these effects already arise for a standard American vanilla put, our findings suggest that model misspecification may have even more pronounced implications for more complex or path-dependent American-style derivatives. Further research is therefore warranted to assess the robustness of exercise policies in richer product classes and market environments.


\section*{Disclosure}

The authors report there are no competing interests to declare.

\section*{Funding}

No funding was received.

\bibliographystyle{abbrvnat}
\setlength{\bibsep}{0pt}
\bibliography{References}

\titleformat{\section}
  {\normalfont\Large\bfseries}
  {Appendix \thesection:}
  {1em}
  {}
  
\appendix

\section{FD Schemes}
\label{sec:AppendixFDSchemes}

In the following, we briefly review different schemes that may be used to approximate derivatives using finite differences. Consider $f: \mathbb{R} \rightarrow \mathbb{R}$ and let $x_0 < x_1 < x_2 \ldots < x_m$ be some given mesh, where we define $\Delta_{x_i}:= x_i - x_{i-1}$. We may then approximate $f'(x_i)$ through 

\begin{equation}
\label{eq:Upwind1FirstDeriv}
f'(x_i) \approx \alpha_{i,-2} f(x_{i-2}) + \alpha_{i,-1} f(x_{i - 1}) + \alpha_{i, 0} f(x_{i}),
\end{equation}
\begin{equation}
\label{eq:CentralFirstDeriv}
f'(x_i) \approx \beta_{i,-1} f(x_{i-1}) + \beta_{i,0} f(x_{i}) + \beta_{i, 1} f(x_{i+1}),
\end{equation}
\begin{equation}
\label{eq:Upwind2FirstDeriv}
f'(x_i) \approx \gamma_{i, 0} f(x_{i}) + \gamma_{i,1} f(x_{i +1 }) + \alpha_{i, 2} f(x_{i + 2}),
\end{equation}

with coefficients given by 

$$
\alpha_{i, - 2} = \frac{\Delta_{x_i}}{\Delta_{x_{i- 1}}(\Delta_{x_{i-1}}+ \Delta_{x_i})}, \qquad \alpha_{i, - 1} = \frac{-\Delta_{x_{i- 1}} - \Delta_{x_{i}}}{\Delta_{x_{i-1}} \Delta_{x_{i}}}, \qquad \alpha_{i, 0} = \frac{\Delta_{x_{i- 1}} + 2\Delta_{x_{i}}}{\Delta_{x_{i}}(\Delta_{x_{i- 1}} + \Delta_{x_{i}})},
$$

$$
\beta_{i, -1} = \frac{-\Delta_{x_{i +1}}}{\Delta_{x_{i}}(\Delta_{x_{i}}+ \Delta_{x_{i + 1}})}, \qquad \beta_{i, 0} = \frac{\Delta_{x_{i+1}} - \Delta_{x_{i}}}{\Delta_{x_{i}} \Delta_{x_{i +1}}}, \qquad \beta_{i, 1} = \frac{\Delta_{x_{i}}}{\Delta_{x_{i+1}}(\Delta_{x_{i}} + \Delta_{x_{i+1}})},
$$

$$
\gamma_{i, 0} = \frac{-2\Delta_{x_{i + 1}} - \Delta_{x_{i + 2}}}{\Delta_{x_{i+1}}(\Delta_{x_{i+1}}+ \Delta_{x_{i+2}})}, \qquad \gamma_{i, 1} = \frac{\Delta_{x_{i+1}} + \Delta_{x_{i+2}}}{\Delta_{x_{i+1}} \Delta_{x_{i+2}}}, \qquad \gamma_{i, 2} = \frac{-\Delta_{x_{i+ 1}}}{\Delta_{x_{i+2}}(\Delta_{x_{i+1}} + \Delta_{x_{i+2}})}.
$$

We refer to the FD scheme given in Equation \eqref{eq:Upwind1FirstDeriv}, \eqref{eq:CentralFirstDeriv} and \eqref{eq:Upwind2FirstDeriv} as the downward, central and upward scheme, respectively. They can easily be derived using Taylor expansion. Given that $f$ is sufficiently often continuously differentiable we have a second-order error. 

To approximate the second derivative $f''(x)$ we may use the following schemes

\begin{equation}
\label{eq:Upwind1SecondDeriv}
f''(x_i) \approx \eta_{i,-2} f(x_{i-2}) + \eta_{i,-1} f(x_{i - 1}) + \eta_{i, 0} f(x_{i}),
\end{equation}
\begin{equation}
\label{eq:CentralSecondDeriv}
f''(x_i) \approx \delta_{i,-1} f(x_{i-1}) + \delta_{i,0} f(x_{i}) + \delta_{i, 1} f(x_{i+1}),
\end{equation}
\begin{equation}
\label{eq:Upwind2SecondDeriv}
f''(x_i) \approx \epsilon_{i, 0} f(x_{i}) + \epsilon_{i,1} f(x_{i +1 }) + \epsilon_{i, 2} f(x_{i + 2}),
\end{equation}

with coefficients given by 

$$
\eta_{i, - 2} = \frac{2}{\Delta_{x_{i- 1}}(\Delta_{x_{i-1}}+ \Delta_{x_i})}, \qquad \eta_{i, - 1} = \frac{-2}{\Delta_{x_{i-1}} \Delta_{x_{i}}}, \qquad \eta_{i, 0} = \frac{2}{\Delta_{x_{i}}(\Delta_{x_{i- 1}} + \Delta_{x_{i}})},
$$

$$
\delta_{i, -1} = \frac{2}{\Delta_{x_{i}}(\Delta_{x_{i}}+ \Delta_{x_{i + 1}})}, \qquad \delta_{i, 0} = \frac{-2}{\Delta_{x_{i}} \Delta_{x_{i +1}}}, \qquad \delta_{i, 1} = \frac{2}{\Delta_{x_{i+1}}(\Delta_{x_{i}} + \Delta_{x_{i+1}})},
$$

$$
\epsilon_{i, 0} = \frac{2}{\Delta_{x_{i+1}}(\Delta_{x_{i+1}}+ \Delta_{x_{i+2}})}, \qquad \epsilon_{i, 1} = \frac{-2}{\Delta_{x_{i+1}} \Delta_{x_{i+2}}}, \qquad \epsilon_{i, 2} = \frac{2}{\Delta_{x_{i+2}}(\Delta_{x_{i+1}} + \Delta_{x_{i+2}})}.
$$

Now assuming $f: \mathbb{R}^2 \rightarrow \mathbb{R}$, one may consider approximating mixed derivatives. For this, we let $x_i, \Delta_{x_i}$ be as above and additionally consider the mesh $y_0 < y_1 < y_2 < \ldots y_n$ with $\Delta_{y_j} := y_j - y_{j - 1}$. We denote by $\hat{\beta}_{i, k}$ the analogous to $\beta_{i, k}$ in Equation \eqref{eq:CentralFirstDeriv} but in the $y$-direction. For the mixed derivative, we may assume that both derivatives are approximated using a central scheme

\begin{equation}
\label{eq:CentralMixed1}
\frac{\partial^2 f}{\partial x\partial y}(x_i, y_j) \approx \sum_{k, l = -1}^1 \beta_{i, j}\hat{\beta}_{j, l} f(x_{i + k}, y_{i + l}),
\end{equation}

or we may mix, such as, for instance

\begin{equation}
\label{eq:CentralMixed2}
\frac{\partial^2 f}{\partial x\partial y}(x_i, y_j) \approx \sum_{l = 0}^2\sum_{k-1}^1 \beta_{i, j}\hat{\gamma}_{j, l} f(x_{i + k}, y_{i + l}),
\end{equation}

where $\hat{\gamma}_{i, k}$ is the analogous to $\gamma_{i, k}$ in Equation \eqref{eq:Upwind2FirstDeriv}.

\section{Details on the Numerics in the One-Dimensional Models}
\label{sec:AppendixNumericsOneDim}

When considering the one-dimensional models, we follow the approach outlined in \citet[p. 12-14]{Andersen1998}, which is adapted to the Dupire model and simplifies in the Black--Scholes model. In a first step, the $(t, \tilde{X})$-plane is discretized into a uniformly spaced mesh with $n_1 + 1$ and $n_2 + 1$ entries, respectively,

$$
t_j = j \Delta_t = j \frac{T}{n_1} \qquad j = 0, \ldots, n_1,
$$

$$
x_i = x_0 + i \Delta_x = x_0 + i \frac{x_{n_2} - x_0}{n_2} \qquad i = 0, \ldots, n_2.
$$

We approximate the derivatives as 

\begin{equation}
\label{eq:DiscretOneDim}
\left.
\begin{aligned}
\frac{\partial \tilde{U}(t_j, x_i)}{\partial t} &\approx \frac{\tilde{U}(t_{j + 1}, x_i) - \tilde{U}(t_j, x_i)}{\Delta_t},\\
\frac{\partial \tilde{U}(t_j, x_i)}{\partial \tilde{X}} &\approx (1 - \lambda_1) \frac{\tilde{U}(t_{j}, x_{i + 1}) - \tilde{U}(t_j, x_{i-1})}{2 \Delta_x} + \lambda_1 \frac{\tilde{U}(t_{j + 1}, x_{i + 1}) - \tilde{U}(t_{j + 1}, x_{i-1})}{2 \Delta_x},\\
\frac{\partial \tilde{U}(t_j, x_i)}{\partial \tilde{X}^2} &\approx (1 - \lambda_1) \frac{\tilde{U}(t_{j}, x_{i + 1}) - 2 \tilde{U}(t_j, x_i) + \tilde{U}(t_j, x_{i-1})}{\Delta_x^2} \\ &\quad +\lambda_1 \frac{\tilde{U}(t_{j + 1}, x_{i + 1}) - 2\tilde{U}(t_{j + 1}, x_i) + \tilde{U}(t_{j + 1}, x_{i-1})}{\Delta_x^2},
\end{aligned}
\right\}
\end{equation}

where the parameter $\lambda_1 \in [0, 1]$ controls the time at which the partial derivatives with respect to the state variable are evaluated. As \citet[p. 13]{Andersen1998} note, for $\lambda_1 = 0$ we have a fully implicit FD scheme, for $\lambda_1 = 1$ we have an explicit scheme, and for $\lambda_1 = 0.5$ we obtain the Crank-Nicolson scheme. \citet[p. 15]{Andersen1998} note that both the implicit and Crank-Nicolson scheme are unconditionally stable as long as the local volatility remains non-negative. 

We point out that a very common technique of solving PDE's is the method-of-lines, in which one first discretizes space, leading to a system of Ordinary Differential Equations (ODEs), and then in a second step discretizes time, where the Crank-Nicolson scheme is one possible choice. In doing so, derivatives with respect to $\tilde{X}$ may be approximated using central schemes, utilizing \eqref{eq:CentralFirstDeriv} and \eqref{eq:CentralSecondDeriv}, after which time is discretized using a weighted forward/backward Euler scheme. This leads to a FD scheme that is not entirely identical to the one suggested by \cite{Andersen1998}, and used by us, since it considers two different values of the local volatility - one evaluated at time $t_j$ and one evaluated at $t_{j + 1}$, and it also distributes $r$ between $\tilde{U}(t_j, x_i)$ and $\tilde{U}(t_{j + 1}, x_i)$ according to the chosen weight. In our scheme, we freeze the coefficient of the local volatility from time-step $t_j$ until $t_{j + 1}$. Further, we assign $r$ only to $U(t_j, x_i)$, i.e. we treat the $r\tilde{U}$ term in an explicit way.

\subsection{Valuation of American Put}

We discretize time by setting $n_1 = 300$. We choose $x_0$ and $x_{n_2}$ such that the current value, $\log(S(0))$, is approximately in the middle of the grid. We set $n_2 = 1001$ and the maximal stock price as $8\times S(0) = 80$, i.e. we consider the grid $x_0 = 0.2272941,\ldots, x_{1001} = \log(80)$. We set $\lambda_1 = 0.5$. 

The discretization \eqref{eq:DiscretOneDim} of the PDE-part of the variational inequality leads to the following tridiagonal system of equations.

\begin{equation}
\label{eq:TridiagSystem}
\left((1 + r\Delta_t) I - (1-\lambda_1) M_j\right) \tilde{U}_j = (I + \lambda_1 M_j)\tilde{U}_{j + 1} + B_j,
\end{equation}

for the $(n_2 - 1) \times (n_2 - 1)$ matrix 

$$
M_j=\begin{bmatrix}
c_{1, j} & u_{1, j} & 0        & 0        & 0      & \ldots     & 0\\
l_{2, j} & c_{2, j} & u_{2, j} & 0        & 0      & \ldots     & 0 \\
0        & l_{3, j} & c_{3, j} & u_{3, j} & 0      & \ldots     & 0  \\
\vdots   & \vdots   & \vdots   & \vdots   & \vdots & \vdots     & \vdots \\
0        & 0        & 0        &          & \ldots & l_{(n_2-1), j} & c_{(n_2-1), j}
\end{bmatrix}
$$

where 

$$
\begin{aligned}
l_{i, j} &= \frac{1}{2} \alpha \left(v(t_j, x_i) - \Delta_x w(t_j, x_i) \right),\\
c_{i, j} &= - \alpha v(t_j, x_i),\\
u_{i, j} &= \frac{1}{2} \alpha \left(v(t_j, x_i) + \Delta_x w(t_j, x_i) \right), \\
\alpha &= \frac{\Delta_t}{\Delta_x^2},
\end{aligned}
$$

and the $(n_2 -1)$ vectors 

$$
\tilde{U}_j=
\begin{bmatrix}
\tilde{U}(t_j, x_1)\\
\tilde{U}(t_j, x_2)\\
\tilde{U}(t_j, x_3)\\
\vdots\\
\tilde{U}(t_j, x_{n_2 - 1})\\
\end{bmatrix}, 
\tilde{U}_{j + 1}=
\begin{bmatrix}
\tilde{U}(t_{j + 1}, x_1)\\
\tilde{U}(t_{j + 1}, x_2)\\
\tilde{U}(t_{j + 1}, x_3)\\
\vdots\\
\tilde{U}(t_{j + 1}, x_{n_2 - 1})\\
\end{bmatrix}
$$
and
$$
B_j = 
\begin{bmatrix}
l_{1, j}\left((1-\lambda_1)\tilde{U}(t_j, x_0) + \lambda_1 \tilde{U}(t_{j + 1}, x_0\right)\\
0 \\
0 \\
\vdots \\
u_{(n_2 - 1), j}\left((1-\lambda_1)\tilde{U}(t_j, x_{n_2}) + \lambda_1 \tilde{U}(t_{j + 1}, x_{n_2}\right)\\
\end{bmatrix},
$$

where $\tilde{U}(t_j, x_0)$ and $\tilde{U}(t_j, x_{n_2})$ are specified through our boundary conditions for any $j = 0, \ldots, n_1 - 1$.

We approximate the boundary conditions of the American put by setting

\begin{itemize}
    \item $\tilde{U}(t_j, x_0) = K - \exp(x_0), \qquad j = 0, \ldots n_1 -1$,
    \item $\tilde{U}(t_j, x_{n_2}) = 0, \qquad j = 0, \ldots n_1-1$.
\end{itemize}

Going backwards in time we solve \eqref{eq:TridiagSystem}, and adapt the FD value whenever the immediate payoff exceeds it (within the iteration). Implicitly, we obtain an estimate of the exercise boundary at each discretized time-step by choosing the maximal stock value in the grid where early exercise takes place. 

\subsection{Calibration of the Dupire model}
\label{sec:AppendixCalibration}
For the calibration of the Dupire model, we use the same discretization of the $(t,\tilde{X})$-plane as for the valuation of the American put in the one-dimensional models. We inter- \& extrapolate the observed market prices using bicubic splines with a natural end condition, where we make use of a two-step spline approach, in which we first fit splines into the $K$-direction, and then we fit them into the $T$-direction. 

In the calibration approach proposed by \cite{Andersen1998}, $\lambda_1$ must lie within $[0, 1)$. We keep $\lambda_1 = 0.5$ and thus work with a semi-implicit approach. \citet[Equation 41, p. 19]{Andersen1998} obtain a discretized version of the Fokker-Planck equation, which can be re-arranged into a form that suggests a quadratic program subject to some bounds for the local volatility \citep[Equation 47, p. 22]{Andersen1998}. Further, \citet{Andersen1998} suggest that one could limit the optimization to the statistically significant region of the grid. We make the volatility bounds very wide and only require $\sigma(t_j, \exp(x_i)) \in [0.01, 6]$, i.e. $v(t_j, x_i) \in [0.0001, 36]$ for any $(t_j, x_i)$ within the optimization grid, which is composed of the discretized time-points that are above the first observed maturity and $\log(2) \leq x_i \leq \log(20)$. We solve the constrained quadratic program in this specified region, and extrapolate outside of the region using our two-step bicubic spline approach, and applying the bounds again if necessary. 

When evaluating the pricing performance of the calibrated local volatility surface, we price European calls by solving \eqref{eq:TridiagSystem} backwards in time and approximating the boundary conditions as

\begin{itemize}
    \item $\tilde{U}(t_j, x_0) = 0, \qquad j = 0, \ldots n_1 -1$,
    \item $\tilde{U}(t_j, x_{n_2}) = \exp(x_{n_2})-K\exp(-r(T-t_j)), \qquad j = 0, \ldots n_1-1$.
\end{itemize}

\section{Details on the Numerics in the Heston Model}
\label{sec:AppendixNumericsHeston}

For creation of the grid for the state variables, we set $m_1 = 500, m_2 = 110$, $s_{\max} = 8\times S(0)$ and $v_{\max} = 4.5$ and $c = K/5, d = v_{\max}/80$, i.e. we have $s_0 = 0, \ldots, s_{500} = 80$ and $v_0 = 0, \ldots, v_{100} = 4.5$. 

Based on our ordering scheme mentioned in Section \ref{sec:AmericanOptionPricingTwoDimNumerics}, the state discretization in the Heston model leads to the following system of stiff ODEs 

\small
$$
\begin{bmatrix}
\frac{\partial V(\hat{t}, s_1, v_0)}{\partial \hat{t}} \\
\vdots \\
\frac{\partial V(\hat{t}, s_{m_1-1}, v_0)}{\partial \hat{t}} \\
\vdots \\
\frac{\partial V(\hat{t}, s_{m_1-1}, v_{m_2-1})}{\partial \hat{t}}
\end{bmatrix}
= 
\begin{bmatrix}
\color{blue}{A^{0, 0}_{1, 1}} & \color{blue}{\ldots} & \color{blue}{A^{0, 0}_{1, m_1-1}} & \ldots &  A^{0, m_2-1}_{1, m_1-1}  \\
\color{blue}{\vdots} & \color{blue}{\ddots} & \color{blue}{\vdots} & \vdots & \vdots \\
\color{blue}{A^{0, 0}_{m_1-1, 1}} & \color{blue}{\ldots} & \color{blue}{A^{0, 0}_{m_1-1, m_1-1}} & \ldots &  A^{0, m_2-1}_{m_1-1, m_1-1}  \\
\vdots & \ddots & \vdots & \vdots & \vdots \\
A^{m_2-1, 0}_{m_1-1, 1} & \ldots & A^{m_2-1, 0}_{m_1-1, m_1-1} & \ldots &  A^{m_2-1, m_2-1}_{m_1-1, m_1-1}
\end{bmatrix}
\cdot\begin{bmatrix}
V(\hat{t}, s_1, v_0)\\
\vdots\\
V(\hat{t}, s_{m_1-1}, v_0)\\
\vdots \\
V(\hat{t}, s_{m_1-1}, v_{m_2-1})
\end{bmatrix}
+
\begin{bmatrix}
b_{1, 0}(\hat{t})\\
\vdots \\
b_{m_1-1, 0}(\hat{t})\\
\vdots \\
b_{m_1-1, m_2-1}(\hat{t})\\
\end{bmatrix},
$$
\normalsize

where the blue matrix is a block matrix. We use the notation $A^{r_1, c_1}_{r_2, c_2}$ where $r_1$ refers to the row index of the block matrix, $c_1$ denotes the column index of the block matrix, $r_2$ denotes the row index in the particular block matrix and $c_2$ the column index in the particular block matrix. Notice that we define A only on the interior points of our grid, which are given by $\{(s_i, v_i): 1 \leq i \leq m_1 - 1, 0 \leq j \leq m_2 - 1\}$. 
The entire matrix A is thus composed of block matrices, where each block matrix is of dimension $(m_1 - 1)\times (m_1 -1)$. In order to approximate $\frac{\partial V(\hat{t}, s_i, v_j)}{\partial \hat{t}}$ for any $i \in \{1, \ldots, m_1-1\}$ and $j \in \{0, \ldots, m_2-1\}$ we usually consider $3$ nonzero block matrices, while the rest of the block matrices in the block-row are zero-matrices. Only in the case where $j$ is such that $v_j > 1$ we consider $4$ nonzero block matrices (unless we are at the final interior point in the $v$-grid, in which case we again only consider $3$ nonzero block matrices). We note that the block matrices are sparse; for instance we have that each of the sub-matrices of $A_0$ and $A_1$ are tridiagonal (or zero-matrices) while $A_2$ is composed of diagonal matrices (and zero-matrices). 

Recalling our boundary conditions, we specify that at $s_{\min}$ we have value $K$, at $s_{\max}$ we have value $0$ and at $v_{\max}$ we again have value $K$. There will be one overlapping of our boundary conditions, namely at $(s_{\max}, v_{\max})$, where we let the boundary condition for $s_{\max}$ win. \\

We discretize time using the same time-grid as for the one-dimensional models, thus setting $m_3 = 300$. We use the notation $V^n$ for the approximation to the exact solution $V(t_n)$. Iteratively, the MCS scheme computes for $n = 1, 2, 3, \ldots, m_3$ 

$
\begin{aligned}
Y_0 &= V^{n-1} + \Delta_t(AV^{n-1} + b^{n-1}) \qquad \text{(Forward Euler)} \\
Y_1 &= (1 - \lambda_2 \Delta_t A_1)^{-1}\left[Y_0 + \lambda_2 \Delta_t (b_1^n - b_1^{n-1}-A_1 V^{n-1}) \right]\\
Y_2 &= (1 - \lambda_2 \Delta_t A_2)^{-1}\left[Y_1 + \lambda_2 \Delta_t(b_2^n - b_2^{n-1} - A_2 V^{n-1}) \right]\\
\hat{Y}_0 &= Y_0 + \lambda_2 \Delta_t \left[A_0(Y_2 - V^{n-1}) + b_0^{n} - b_0^{n-1} \right] \\
\tilde{Y}_0 &= \hat{Y}_0 + \left(\frac{1}{2}-\lambda_2\right) \Delta_t\left[A(Y_2 - V^{n-1}) + b^{n} - b^{n-1} \right]\\
\tilde{Y}_1 &= (1 - \lambda_2 \Delta_t A_1)^{-1}\left[\tilde{Y}_0 + \lambda_2 \Delta_t (b_1^n - b_1^{n-1}-A_1V^{n-1}) \right]\\
\tilde{Y}_2 &= (1 - \lambda_2 \Delta_t A_2)^{-1}\left[\tilde{Y}_1 + \lambda_2 \Delta_t (b_2^n - b_2^{n-1} - A_2V^{n-1}) \right] \\
V^n &= \tilde{Y}_2 \\
\end{aligned}
$

\begin{remark}
Recall that solving this \textit{forward} in \textit{time to maturity} is equivalent to solving it \textit{backward} in \textit{time}. Further, note that since we have a constant boundary vector for the American put, numerous terms within the scheme cancel out. 
\end{remark}

\begin{remark}
We highlight that the matrices $A_0, A_1, A_2$ are constant, and thus the inverses that are used above have to be computed once only, and can then be re-used for all iterations, which significantly speeds up computations. 
\end{remark}

In order to account for early exercise of the American put, we adapt this scheme for the pricing; specifically we set

\[
V^n = \max\{\tilde{Y}_2, (K -\hat{s})^+ \}, \qquad \text{ (Early Exercise)} \\
\]

where $\hat{s}$ is the interior part of the repetition of the $S$-grid for each $v$-grid value, i.e. we have the following structure $\hat{s} = (s_1, \ldots, s_{m_1-1}, s_1, \ldots, s_{m_1-1}, s_{m_1-1}, s_1,\ldots s_{m_1-1})$. 

Implicitly, we obtain the optimal exercise boundary for each $v$-grid point by searching for the highest $S$-grid point where early exercise takes place. We interpolate the obtained boundary into the $v$-direction using cubic splines with a natural end condition, where we consider a uniformly spaced grid of $v$-values ranging from $0$ to $2$ with $0.01$ increments. 

\section{Longstaff-Schwartz (Base Case)}
\label{sec:AppendixLSBase}

In the following, we consider the same base case as in Section \ref{sec:NumericalResultsBase}. We employ the Longstaff-Schwartz (LS) algorithm to approximate the optimal exercise boundaries, and make use of a randomization approach originally developed in \citet{Carr1998}, where we randomize the starting value of the stock and let $S(0) \sim N(10, 2.5^2)$ in all models. We simulate $1$ mio. paths with $300$ time-steps under all three models using such a random starting value of the stock. In the Black--Scholes and Dupire model we consider $\{S, S^2, S^3, S^4, S^5, S^6, S^7, S^8\}$ as the set of basis functions and for the Heston model we consider $\{S, S^2, S^3, v, Sv, S^2v\}$. 

Figure \ref{fig:BoundariesWithLS} illustrates the approximate boundaries. The LS approximation is quite similar to the FD approximation when considering the Black--Scholes and Dupire model. Under the Heston model, we notice that the approximations diverge to a larger extent, which is to be expected to some degree as the root-solving becomes more challenging once there is more than one state variable. We notice that the approximation under LS becomes quite non-smooth when considering high volatility values. An application of these approximated exercise boundaries thus leads to noisier results. We omit further details.

\begin{figure}[h!]
\centering
\includegraphics[scale=0.50]{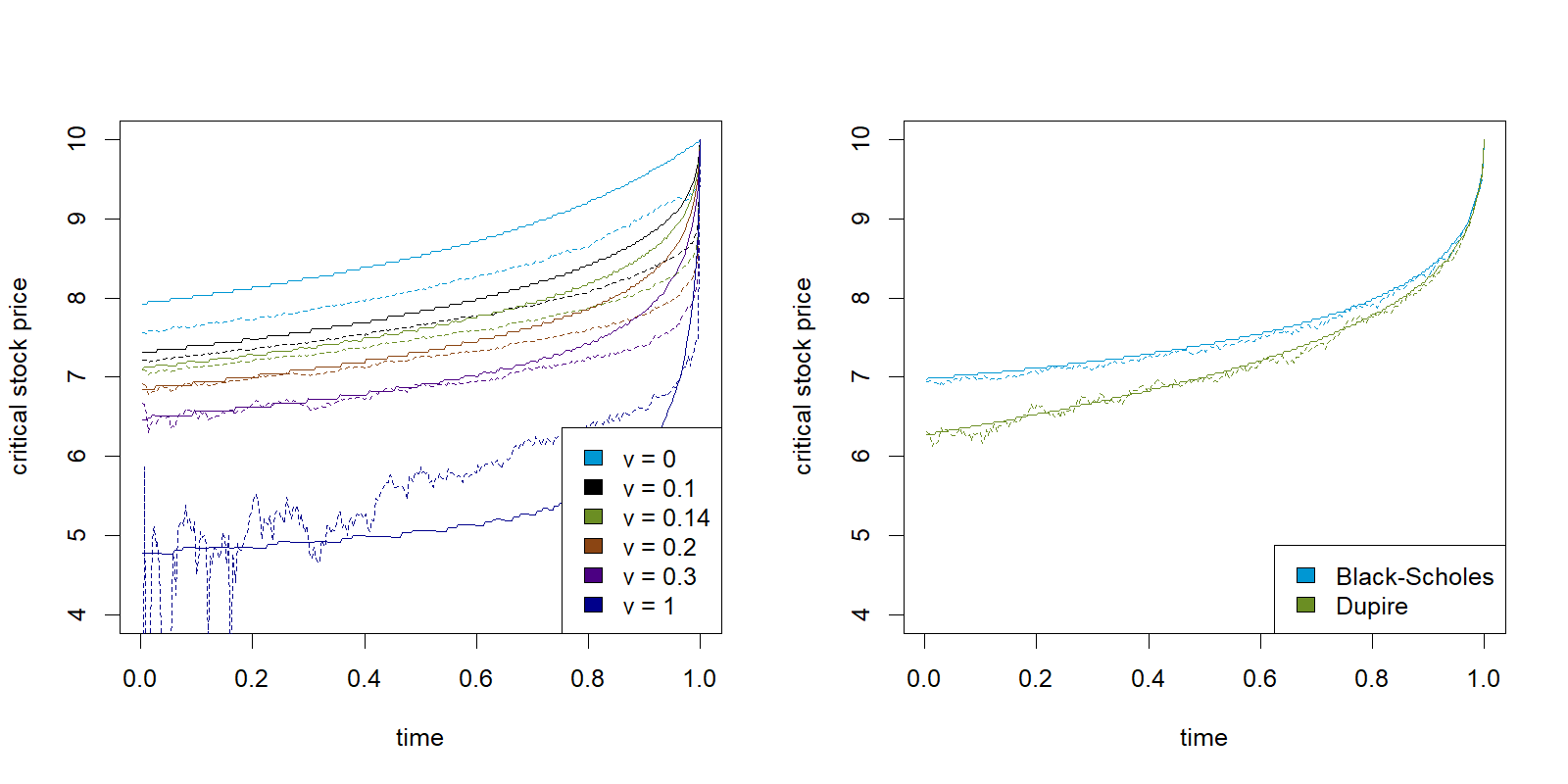}
\caption{Approximate optimal exercise boundaries based on LS (dashed lines) and FD (solid lines). Left: approximate boundary for selected volatility values under the Heston model. Right: approximate boundaries under the misspecified models (Black--Scholes/Dupire).}
\label{fig:BoundariesWithLS}
\end{figure}

\end{document}